\documentclass[reqno]{amsproc}
\usepackage[colorlinks=true,linkcolor=black,citecolor=black,urlcolor=black,pdfstartview=FitH]{hyperref}
\usepackage{amssymb}
\usepackage[english]{babel}
\usepackage{graphicx,subfig}
\pdfoutput=1

\numberwithin{equation}{section}
\allowdisplaybreaks[4]

\newtheorem{prop}{Proposition}[section]
\newtheorem{thm}[prop]{Theorem}
\newtheorem{lemma}[prop]{Lemma}
\newtheorem{cor}[prop]{Corollary}
\newtheorem{df}[prop]{Definition}
\newtheorem{rem}[prop]{Remark}
\newtheorem{ex}[prop]{Example}

\newcommand{\A}{\mathcal{A}}
\newcommand{\B}{\mathcal{B}}
\newcommand{\HH}{\mathcal{H}}
\newcommand{\mc}{\mathcal}
\newcommand{\R}{\mathbb{R}}
\newcommand{\C}{\mathbb{C}}
\newcommand{\N}{\mathbb{N}}
\newcommand{\de}{\mathrm{d}}
\newcommand{\inner}[1]{\left<\smash[t]{#1}\right>}
\newcommand{\ket}[1]{\left|\smash[t]{#1}\right>}

\newcommand{\tr}{\mathrm{Tr}}
\newcommand{\id}{\mathsf{id}}
\newcommand{\ran}{\mathrm{range}}

\title{Pythagoras Theorem in Noncommutative Geometry}

\author{Francesco D'Andrea}

\address{Dipartimento di Matematica e Applicazioni, Universit\`a di Napoli ``Federico II'' and I.N.F.N. Sezione di Napoli, Complesso MSA, Via Cintia, 80126 Napoli, Italy}

\email{francesco.dandrea@unina.it}

\thanks{%
\textit{Date:} September 2015. \indent
\textit{Acknowledgements.}~Research supported by UniNA and Compagnia di San Paolo in the framework of the Program STAR 2013. \vspace{3pt}}

\begin{document}

\begin{abstract}
After a review of the results in \cite{DM13} about Pythagorean inequalities for products of spectral triples, I will present some new results and discuss classes of spectral triples and states for which equality holds.
\end{abstract}

\maketitle


\section{Introduction}
The Pythagorean Theorem (Euclid I.47) is probably the most famous statement in all of mathematics, and also the one with the largest number of proofs.
There is a collection of 366 different proofs is in the book by E.S.~Loomis \cite{Loo68}, and a track online of some false proofs \cite{Bog15}, which includes some from Loomis' book itself.
A celebrated visual proof is the one in Figure~\ref{fig:1a}, which Proclus (ca.~412--485 \textsc{AD}) attributes to Pythagoras himself \cite[p.~61]{Mao07}.
A similar one \cite[Proof~36]{Loo68} is considered the first ever recorded proof of the Pythagorean Theorem \cite[Chap.~5]{Mao07}.
A chronology can be found in \cite[p.~241]{Mao07}.

\begin{figure}[t]
  \hspace*{-7mm}%
  \subfloat[\label{fig:1a}The blue area equals the total area minus four times the red area, so:
	\mbox{$\;\;a^2=(b+c)^2-4(bc/2)=b^2+c^2$.}]
	{\makebox[7.5cm]{\includegraphics[width=5cm]{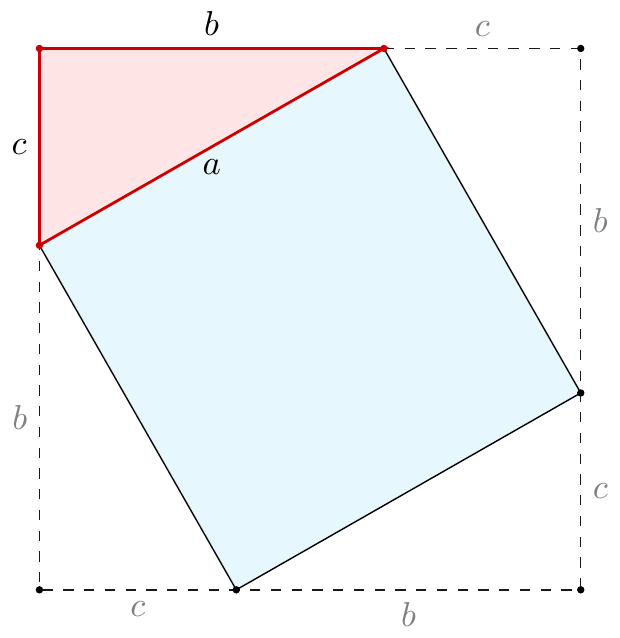}}}
	\hspace*{-8mm}%
	\subfloat[\label{fig:1b}A hyperbolic right-angled triangle]
  {\makebox[7cm]{\includegraphics[width=5cm]{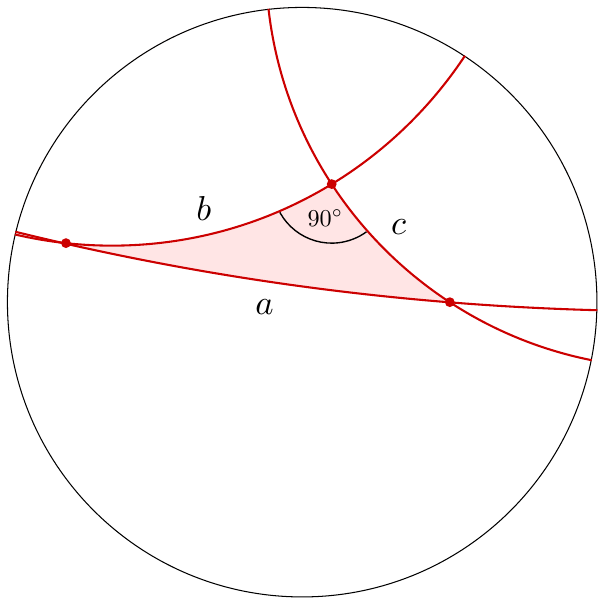}}}
	
	\caption{}
\end{figure}

There are generalizations to hyperbolic and elliptic geometry as well.
For a right-angled geodesic triangle in the hyperbolic plane (Fig.~\ref{fig:1b}): 
$$
\cosh(a)=\cosh(b)\cosh(c) \;,
$$
see e.g.~\cite[p.~81]{Thu97}, while on a unit sphere:
$$
\cos(a)=\cos(b)\cos(c) \;.
$$
In both cases, if one performs a Taylor series expansion, at the leading order one retrieves the usual Pythagoras theorem, which then holds for ``very small triangles''.

\medskip

Pythagoras theorem for ``very small triangles'' is exactly the way one defines the product metric in Riemannian geometry. On a Cartesian product $M=M_1\times M_2$ of two Riemannian manifolds, one defines the line element as
\begin{equation}\label{eq:lineel}
ds^2=ds_1^2+ds_2^2
\end{equation}
where $ds_1$ resp.~$ds_2$ are the line elements on $M_1$ resp.~$M_2$.
As shown in \cite{DM13}, one can integrate last equality, and prove that
Pythagoras equality holds for geodesic right-angled triangles, provided
the two legs are ``parallel'' one to $M_1$ and the other to $M_2$
(cf.~\S\ref{sec:Hodge}).
The first complications arise in this example when points are replaced by probability measures, and the geodesic distance is replaced by its natural generalization: the Wasserstein distance of order $1$ (see e.g.~\cite{Vil08}). One can realize via a simple example that Pythagoras equality is replaced by the inequalities
\begin{equation}\label{eq:Pineq}
\sqrt{b^2+c^2}\leq a\leq \sqrt{2}\,\sqrt{b^2+c^2} \;,
\end{equation}
and that the upper bound -- with coefficient $\sqrt{2}$ -- is optimal (in \S3.3 of \cite{DM13} we exhibit an example where $a$ assumes all values between the lower and upper bound).

\medskip

Now, \eqref{eq:Pineq} is a particular instance of a formula which holds in noncommutative geometry. Given a $C^*$-algebra $A$, we can define a distance on its state space $\mc{S}(A)$ by means of a spectral triple. Suppose $A=A_1\otimes A_2$ is a tensor product of two \mbox{$C^*$-algebras}.
Identifying elements $\varphi=(\varphi_1,\varphi_2)$ and $\psi=(\psi_1,\psi_2)\in\mc{S}(A_1)\times\mc{S}(A_2)$ with product states of $A$, one may wonder whether Pythagoras equality holds for the triangle with vertices $\varphi,\psi$ and $(\psi_1,\varphi_2)$ (Fig.~\ref{fig:2}). It turns out that \eqref{eq:Pineq} is valid even in this more general framework \cite{DM13}. More precisely, the right inequality is valid for arbitrary spectral triples, while the left one is valid for unital spectral triples only, but it holds for arbitrary states (with components $\varphi_1,\varphi_2,\psi_1,\psi_2$ replaced by marginals). I~will present a slightly different proof of the inequalities \eqref{eq:Pineq} for spectral triples in \S\ref{sec:prodsp} (including the generalization to marginals).

\begin{figure}[t]
  \centering\hspace*{3mm}%
	\includegraphics[height=4.8cm]{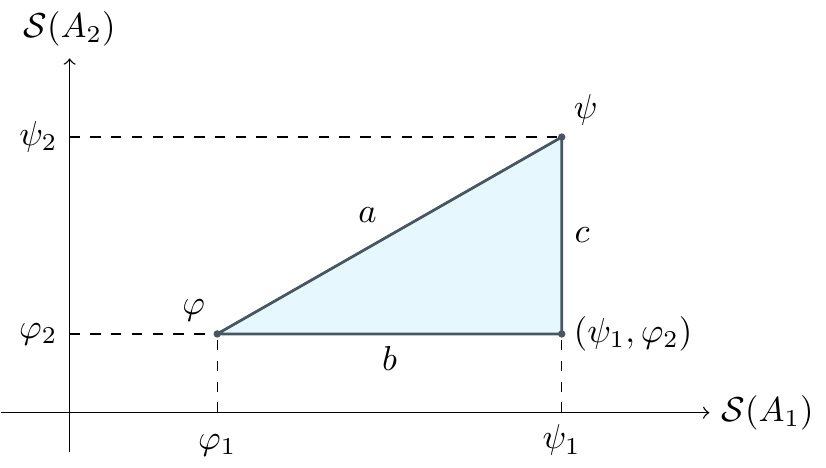}
	
	\vspace{-10pt}
	
	\caption{}
	\label{fig:2}
\end{figure}

It is then natural to wonder for which states and spectral triples the left inequality in \eqref{eq:Pineq} is actually an equality, that is: Pythagoras theorem holds.
In the example $M=M_1\times M_2$ of Riemannian manifolds, it holds for geodesic triangles, hence when the vertices are pure states of the commutative algebra $C_0(M_1)\otimes C_0(M_2)$~\cite{DM13}.

In the case of pure states, the equality was proved for the product of any Riemannian manifold with a two-point space \cite{MW02}, and for the product of Moyal plane and a two-point space \cite{MT13} under the hypothesis that the two states on Moyal plane are obtained one from the other by a translation (so, even in this non-commutative example, there is a geodesic flow connecting the vertices of the triangle).

In \S\ref{sec:CN}, I will discuss the discrete analogue of the example of Riemannian manifolds: I will prove that the equality holds for pure states in the product of canonical spectral triples associated to two arbitrary \emph{finite} metric spaces.
One interest for this proof is that it is purely algebraic and makes no use of geodesics.
In \S\ref{sec:geo}, I will generalize the result in
\cite{MW02,MT13} to the product of an arbitrary spectral triple with the two-point space, under the hypothesis that the states considered are connected by some generalized geodesic; I will then explain how to extend such a result from the two-point space to an arbitrary finite metric space.
As a corollary, we get Pythagoras equality for arbitrary pure states in the product of a Riemannian manifold and a finite metric space, thus completing the ``commutative'' picture, and the product of Moyal plane with a finite metric space.

At this point one could believe that Pythagoras equality is a property of pure states, but there is at least one example where it holds for \emph{arbitrary} product states: it is the product of two two-point spaces (the simplest example conceivable), as I will explain in \S\ref{sec:C2}.

A much more difficult question is whether the equality holds for pure states in the product of arbitrary (unital\footnote{In the non-unital case, sometimes even the left inequality in \eqref{eq:Pineq} is violated by pure states.}) spectral triples. In fact, pure states of a commutative \mbox{$C^*$-algebra} are characters. At the present time, there is no general proof nor counterexamples (even noncommutative) to Pythagoras equality for pure/character states.

\medskip

The paper is organized as follows: there is a first part discussing the general theory (\S\ref{sec:2} and \S\ref{sec:prodsp}), a second part where I collected examples (from \S\ref{sec:4} to \S\ref{sec:geo}), and a section (\S\ref{sec:pur}) where I highlight some analogies between the spectral distance and the purified distance of quantum information.

\section{Some preliminary definitions}\label{sec:2}

\subsection{Cartesian products and the product metric}\label{sec:2.1}

If $X$ is a set, we call $d:X\times X\to [0,+\infty]$ an \emph{extended
semi-metric} if, for all $x,y,z\in X$: \\[5pt]
\begin{tabular}{cp{5cm}l}
i)   & $d(x,y)=d(y,x)$ & (symmetry); \\[2pt]
ii)  & $d(x,x)=0$ & (reflexivity); \\[2pt]
iii) & $d(x,y)\leq d(x,z)+d(z,y)$ & (triangle inequality).
\end{tabular} \\[5pt]
If in addition \\[5pt]
\begin{tabular}{cp{5cm}l}
iv)  & $d(x,y)=0\;\Rightarrow\;x=y$ & (identity of the indiscernibles),
\end{tabular} \\[5pt]
we call $d$ an \emph{extended metric} \cite{DD09}.
It is a \emph{metric} `tout court' if $d(x,y)<\infty\;\forall\;x,y$.

\medskip

For $p>0$, we denote by
$$
\|v\|_p=\big(|v_1|^p+|v_2|^p+\ldots+|v_n|^p\big)^{\frac{1}{p}}
$$
the $p$-norm of $v=(v_1,\ldots,v_n)\in\R^n$, and by $\|v\|_\infty=\max_i|v_i|$ the sup (max) norm.

\medskip

Given two metric space $(X_1,d_1)$ and $(X_2,d_2)$, one can verify that  the formula
$$
d(x,y):=\left\|\big(\,d_1(x_1,y_1)\,,\,d_2(x_2,y_2)\,\big)\right\|_p
 \;,\quad \forall\;x=(x_1,x_2),
y=(y_1,y_2)\in X_1\times X_2,
$$
defines a metric on the Cartesian product $X=X_1\times X_2$, for any fixed $p\in (0,\infty]$. For $p=2$, we call this distance the \emph{product metric} and denote it by $d_1\boxtimes d_2$.

More generally (for any $p$), if $d_1$ and $d_2$ are extended metrics, then $d$ is an extended metric, and if they are extended semi-metrics, so is $d$.\footnote{The non-trivial part is to prove that $d$ satisfies the triangle inequality. For $x,y,z\in X$, let $u,v,w\in\R^2$ be the vectors $u=\big(d_1(x_1,y_1),d_2(x_2,y_2)\big)$,
$v=\big(d_1(x_1,z_1),d_2(x_2,z_2)\big)$, and \mbox{$w=\big(d_1(z_1,y_1),d_2(z_2,y_2)\big)$}.
Then $d(x,y)=\|u\|_p\leq\|v+w\|_p$ (since the norm is a non decreasing function of the components, and $d_1,d_2$ satisfy the triangle inequality). From the norm inequality $\|v+w\|_p\leq\|v\|_p+\|w\|_p=d(x,z)+d(z,y)$ we get then the triangle inequality for $d$.}

\subsection{Forms and states on a \texorpdfstring{$C^*$-algebras}{C*-algebras}}
A \emph{form} on a complex vector space $A$ is a linear map $\varphi:A\to\C$. Given a form $\varphi:A\to\C$ on a unital algebra $A$, I will denote by $\varphi^\sharp:A\to A$ the map $\varphi^\sharp(a)=\varphi(a)1_A$.
Note that if $\varphi(1_A)=1$, then $\varphi^\sharp$ is an \emph{idempotent}.
Here by idempotent I mean an endomorphism $\pi$ of a vector space
satisfying $\pi\circ\pi=\pi$, while I will reserve the term \emph{projection} for bounded operators $p$ on a Hilbert space satisfying $p^2=p^*=p$.

\smallskip

A \emph{state} on a complex $C^*$-algebra $A$ is a form which is positive and normalized:
$$
\varphi(a^*a)\geq 0\;\forall\;a\in A \;,\qquad
\|\varphi\|=1 \;,
$$
where $\|\varphi\|$ is the norm dual to the $C^*$-norm of $A$:
$$
\|\varphi\|:=\sup_{a\in A,a\neq 0}\frac{\|\varphi(a)\|}{\|a\|}
$$
If $A$ is unital, the normalization condition is equivalent to $\varphi(1_A)=1$.

States of $A$ form a convex set $\mc{S}(A)$. Its extremal points -- states that cannot be written as a convex combination of other (two or more) states -- are called \emph{pure}. Pure states $\hat{x}$ of $C_0(M)$, with $M$ a locally compact Hausdorff space, are identified with points via the formula $\hat{x}(f):=f(x)$, and are in fact \emph{characters} (one dimensional representations) of the algebra $C_0(M)$.

If $A$ is a concrete $C^*$-algebra of bounded operators on a Hilbert space $\HH$, a \emph{normal} state $\varphi$ is one of the form
$$
\varphi(a)=\tr_{\HH}(\rho\hspace{1pt}a) \;,\qquad\forall\;a\in A,
$$
where $\rho$ is a positive trace class operator normalized to $1$, called a \emph{density operator} or \emph{density matrix} for $\varphi$ (and in general is not unique).
For an abstract $C^*$-algebra, we can talk about normal states with respect to a given representation (cf.~\S2.4.3 and Def.~2.4.25 of \cite{BR97}).

If $A$ is a finite-dimensional $C^*$-algebra every state is normal; if $A=M_n(\C)$ with its natural representation on $\C^n$ the correspondence between density matrices and states is a bijection; this is true even if $A$ is the $C^*$-algebra of compact operators on a infinite-dimensional separable Hilbert space, with the weak$^*$ topology induced by the trace norm on density matrices (Prop.~2.6.13
and Prop.~2.6.15 of \cite{BR97}).

For a discussion of states from the point of view of quantum physics one can also see \cite{Lan98book}. On the probabilistic interpretation of states, and relation between quantum mechanics and probability theory, one can see \cite[App.~A]{RP11}.

\subsection{States of a bipartite system}
Let $A=A_1\otimes A_2$ be the minimal tensor product of two unital $C^*$-algebras $A_1$ and $A_2$.
Given a state $\varphi$ on $A$, we define its \emph{marginals} as the states $\varphi_1^\flat\in\mc{S}(A_1)$ and $\varphi_2^\flat\in\mc{S}(A_2)$ given by\footnote{To simplify the discussion, I give here the definition only in the unital case.}
$$
\varphi_1^\flat(a_1)=\varphi(a_1\otimes 1) \;,\qquad
\varphi_2^\flat(a_2)=\varphi(1\otimes a_2) \;,
$$
for all $a_1\in A_1$ and $a_2\in A_2$.

We call $\varphi$ a \emph{product state} if it is the tensor product of its marginals\footnote{Here I adopt the terminology of \cite{BZ06,RP11} (slightly different from the one of \cite{DM09}).}, i.e.~if it is of the form $\varphi=\varphi_1\otimes\varphi_2$ for some $\varphi_1\in\mc{S}(A_1)$ and $\varphi_2\in\mc{S}(A_2)$ (note that the latter condition makes sense in the non-unital case too).
The set of product states will be identified with the Cartesian product \mbox{$\mc{S}(\A_1)\times\mc{S}(\A_2)$} via the bijection $\varphi_1\otimes\varphi_2\mapsto (\varphi_1,\varphi_2)$, and more generally there is a surjective (but not injective) map
\begin{equation}\label{eq:margmap}
\mc{S}(A)\twoheadrightarrow\mc{S}(A_1)\times\mc{S}(A_2)
\;,\qquad
\varphi\mapsto (\varphi_1^\flat,\varphi_2^\flat)
\;.
\end{equation}
A state is called \emph{separable} if it is a convex combination of product states. In fact, for infinite-dimensional algebras, we may want to consider the \emph{closed} convex hull of product states, and call $\varphi\in\mc{S}(A)$ separable if it is of the form
\begin{equation}\label{eq:separable}
\varphi=\sum_{n\geq 0}p_n\,\varphi_{1,n}\otimes\varphi_{2,n} \;,
\end{equation}
where $\varphi_{i,n}\in\mc{S}(A_i)$ for all $i=1,2$ and $n\geq 0$,
$\{p_n\}_{n\geq 0}$ is a probability distribution, and the series is convergent in the weak$^*$ topology. A state which is not separable is called \emph{entangled}.

In the case of matrix algebras, separable states are the ones with density matrix of the form
$$
\rho=\sum_{i=1}^kp_i\,\rho_{1,i}\otimes\rho_{2,i} \;.
$$
If $\rho$ is of the above type, applying the transposition to the first factor we get a new positive matrix $\rho^{T_1}$. This simple observation (called ``Peres' criterion'', cf.~\S15.4 of \cite{BZ06}) allows to prove that entangled states exist.

\begin{ex}\label{ex:2.1}
Let $A_1=A_2=M_2(\C)$, $\{e_i\}_{i=1,2}$ be the canonical basis of $\C^2$ and $e_{ij}$ the matrix with $1$ in position $(i,j)$ and zero everywhere else.
Let $\rho_{\,\mathrm{Bell}}$ be the projection in the direction of the unit vector
$\frac{1}{\sqrt{2}}(e_1\otimes e_1+e_2\otimes e_2)$, hence
$$
\rho_{\,\mathrm{Bell}}=\frac{1}{2}\sum\nolimits_{ij}e_{ij}\otimes e_{ij} \;.
$$
Thus, $\rho_{\,\mathrm{Bell}}$ is a pure state of the composite system $A=A_1\otimes A_2$. Since
$$
\rho_{\,\mathrm{Bell}}^{T_1}=\frac{1}{2}\sum\nolimits_{ij}e_{ij}\otimes e_{ji}
$$
has an eigenvector $v:=e_1\otimes e_2-e_2\otimes e_1$ with negative eigenvalue ($\rho_{\,\mathrm{Bell}}^{T_1}v=-\frac{1}{2}\,v$), it is not a positive matrix, and the state associated to $\rho_{\,\mathrm{Bell}}$ is entangled.
\end{ex}

The study of entanglement is an extremely active area of research in physics. For the quantum mechanical interpretation of marginals as projective measurements, and entanglement in quantum mechanics, one can see \cite[\S15]{BZ06} or \cite[\S10.2]{RP11}.
The one in Example \ref{ex:2.1} is a composite system of two qubits, and $\rho_{\,\text{Bell}}$ is a Bell states for this system \cite[Ex.~3.2.1]{RP11}.

\medskip

Note that any separable pure states is a product state, since a pure state by definition cannot be a convex combination of other two or more states. Hence for pure states these two notions coincide.
On the other hand, pure states can be entangled, as in Example \ref{ex:2.1}.

In the case of commutative $C^*$-algebras $A_1=C_0(X_1)$ and $A_2=C_0(X_2)$, from the identification $A=A_1\otimes A_2\simeq C_0(X_1\times X_2)$ we deduce that every pure state of $A$ ($=$ point of $X_1\times X_2$) is a product state.
In the unital case, $\mc{S}(A)$ is a compact convex set, hence by
Krein-Milman theorem it is the closed convex hull of its extreme points. This and the above observation on pure states proves that every state of a commutative unital $C^*$-algebra $A$ is separable: entanglement is an exclusive property of noncommutative $C^*$-algebras/quantum systems.

\begin{ex}
If $A=M_n(\C)$, pure states are in bijection with rank $1$ projections, and then with points of $\C\mathrm{P}^{n-1}$.
The map $\C\mathrm{P}^{m-1}\times\C\mathrm{P}^{n-1}\to\C\mathrm{P}^{mn-1}$
into product states of a bipartite system is given by the Segre embedding \cite[\S4.3]{BZ06}.
\end{ex}

\begin{ex}
If $A=\C^n$, $\mc{S}(A)$ is the set of probability distributions on $n$ points, which geometrically is the standard $(n-1)$-simplex:
$$
\mc{S}(A)=\Delta^{n-1}:=\big\{\,
p=(p_1,\ldots,p_n)\in\R^n\,:\,p_i\geq 0\;\wedge\;\textstyle{\sum_i}\,p_i=1\,
\big\} \;.
$$
The embedding as product states $\Delta^{m-1}\times\Delta^{n-1}\to\Delta^{mn-1}$, $(p,q)\mapsto x:=(p_iq_j)$, has for image
the subset of $x=(x_{ij})$ satisfying the algebraic equations
$$
x_{ij}x_{kl}=x_{il}x_{kj} \;,\qquad\forall\;i,j,k,l.
$$
If $m=n=2$, we get the quadric surface described in Appendix \ref{app:A}.
\end{ex}

\subsection{Spectral triples and the spectral distance}
If $\A$ is a $*$-algebra of bounded operators on a complex Hilbert space, we denote by $\mc{S}(\A)$ the set of states of the norm closure of $\A$, and by $\A^{\mathrm{sa}}$ the set of selfadjoint elements of $\A$.

A natural way to construct a metric on $\mc{S}(\A)$ is by means of a spectral triple. Standard textbooks on this subject are \cite{Con94,GVF01,Lan02}.

\begin{df}
A \emph{spectral triple} $(\A,\HH,D)$ is the datum of a $*$-algebra of bounded operators on a Hilbert space $\HH$ and a (unbounded) selfadjoint operator $D$ on $\HH$, such that $a(D+i)^{-1}$ is compact and $[D,a]$ is bounded for all
$a\in\A$. A spectral triple is:
\begin{list}{}{\leftmargin=1.5em \itemsep=1pt}
\item[i)] \emph{unital} if $1_{\B(\HH)}\in\A$ (the algebra is unital and its unit is the identity on $\HH$);
\item[ii)] \emph{even} if there is a grading $\gamma=\gamma^*$, $\gamma^2=1$, commuting with $\A$ and anticommuting with $D$.
\end{list}
We call $D$ a (generalized) Dirac operator.
\end{df}

Although the definition makes sense for real Hilbert spaces as well, and there are examples where one is forced to work over the field of real numbers (e.g.~in the spectral action approach to the internal space of the Standard Model of elementary particles, see e.g.~\cite{CM08,SuijBook}), here I will focus on complex algebras and spaces (the only exception being the example in \S\ref{sec:Bloch}).

\medskip

If $(\A,\HH,D)$ is a spectral triple, an extended metric $d_D$ on $\mc{S}(\A)$
is defined by:
\begin{equation}\label{eq:spectrald}
d_D(\varphi,\psi)=\sup_{a\in\A^{\mathrm{sa}}}\big\{\varphi(a)-\psi(a):\|[D,a]\|\leq 1\big\}\;,
\end{equation}
for all $\varphi,\psi\in\mc{S}(\A)$. We refer to this as the \emph{spectral distance}.

\medskip

Note that from any odd spectral triple $(\A,\HH,D)$ we can construct an even spectral triple \mbox{$(\A,\HH\otimes\C^2,D\otimes\sigma_1,1\otimes\sigma_3)$}
without changing the spectral distance. Here $\sigma_1,\sigma_3$ are Pauli matrices, we use the obvious representation $a\mapsto a\otimes 1$ of $\A$ on $\HH\otimes\C^2$, and since $[D,a]$ and $[D\otimes\sigma_1,a\otimes 1]=[D,a]\otimes\sigma_1$ have the same norm, clearly the distance doesn't changes. So, we do not loose generality by considering only even spectral triples.

\subsection{Spectral distance between normal states}
For normal states there is a formula for the distance \eqref{eq:spectrald} which is a little bit more explicit.
Let $(\A,\HH,D)$ be a spectral triple and
$\varphi_1,\varphi_2$ two distinct normal states with density matrices $\rho_1,\rho_2$ satisfying $\rho:=\rho_1-\rho_2\in\A$.
Call
\begin{equation}\label{eq:innerHS}
\inner{a,b}_{\mathrm{Tr}}=\tr(a^*b)
\end{equation}
and let $\|a\|_{\mathrm{Tr}}$ be the associated norm (whenever they are well-defined, for example for
$a\in\mc{L}^1(\HH)$ and $b\in\B(\HH)$ the former, and for $a\in\mc{L}^2(\HH)$ the latter). Let
$$
V_\rho:=\big\{a=a^*\in\A:\inner{a,\rho}_{\mathrm{Tr}}=0\big\}
$$
be the subspace of $\A^{\mathrm{sa}}$ orthogonal to $\rho$,
and set
$$
L(\rho):=\inf_{a\in V_\rho}\|[D,\rho+a]\| \;.
$$

\begin{prop}\label{prop:drho}
Either $L(\rho)=0$ and $d_D(\varphi_1,\varphi_2)=\infty$, or $L(\rho)\neq 0$ and
\begin{equation}\label{eq:drho}
d_D(\varphi_1,\varphi_2)=L(\rho)^{-1}\|\rho\|_{\mathrm{Tr}}^2 \;.
\end{equation}
\end{prop}
\begin{proof}
Let $W:=\{a\in\A^{\mathrm{sa}}:a\notin V_\rho\}$.
We distinguish between two cases:
\begin{itemize}\itemsep=2pt
\item[i)] $\exists\;a_0\in W$ such that $[D,a_0]=0$;
\item[ii)] $[D,a]\neq 0$ for all $a\in W$.
\end{itemize}
In the first case, any $a=\lambda a_0$ satisfies $\|[D,a]\|\leq 1$ and:
$$
d_D(\varphi_1,\varphi_2) \geq \lambda\inner{\rho,a_0}_{\mathrm{Tr}} \qquad\forall\;\lambda\in\R\;.
$$
Since $\inner{\rho,a_0}_{\mathrm{Tr}}\neq 0$ by definition of $W$, we get $d_D(\varphi_1,\varphi_2)=\infty$. Moreover, called
$$
b=ta_0-\rho \;,\qquad t:=\frac{\|\rho\|_{\mathrm{Tr}}^2}{\inner{\rho,a_0}_{\mathrm{Tr}}} \;,
$$
one checks that $b\in V_\rho$, and $[D,\rho+b]=t[D,a_0]=0$. Therefore, $L(\rho)=0$. Thus, the proposition is proved in case (i). Now we pass to case (ii).

\smallskip

Since $\inner{\rho,a}_{\mathrm{Tr}}=0$ for $a\in V_\rho$, we can write
$$
d_D(\varphi_1,\varphi_2) =
\sup_{a\in W}\frac{ \inner{\rho,a}_{\mathrm{Tr}} }{ \|[D,a]\| } \;,
$$
where by hypothesis the denominator is not zero.

Any $a\in\A^{\mathrm{sa}}$, $a\notin V_\rho$, can be written in a unique way as
$$
a=\lambda(\rho+b) \;,
$$
with $\lambda\neq 0$ and $b\in V_\rho$
(take $\lambda:=\inner{\rho,a}_{\mathrm{Tr}}/\|\rho\|_{\mathrm{Tr}}^2$, $b:=\lambda^{-1}a-\rho$, and check that $b\in V_\rho$). Thus:
$$
d_D(\varphi_1,\varphi_2) =\sup_{\lambda\neq 0,b\in V_\rho}
\frac{\lambda\|\rho\|_{\mathrm{Tr}}^2}{\|\lambda[D,\rho+b]\|}
=\sup_{b\in V_\rho}
\frac{\|\rho\|_{\mathrm{Tr}}^2}{\|[D,\rho+b]\|} \;.
$$
Now the conclusion is obvious: either the sup is infinite, when
the inf of the denominator is zero, or $L(\rho)\neq 0$ and
$d_D(\varphi_1,\varphi_2)=L(\rho)^{-1}\|\rho\|_{\mathrm{Tr}}^2$.
\end{proof}

\begin{rem}
In \cite[Eq.~(38)]{SC13}, the authors give a formula similar to \eqref{eq:drho}, with $L$ replaced by $\|[D,\rho]\|$. They claim that
$d_D(\varphi_1,\varphi_2)$ is equal to:
$$
\|\rho\|_{\mathrm{Tr}}^2 / \|[D,\rho]\| \;.
$$
This statement unfortunately is wrong.
A counterexample is in Prop.~3.10 of \cite{CDMW09}: we have normal pure states at infinite distance, despite the fact that $[D,\rho]$ is not zero (only scalar multiples of the identity commute with the Dirac operator of \cite{CDMW09}).
\end{rem}

\begin{rem}\label{rem:2.7}
Let $\A=M_n(\C)$ be naturally represented on $\HH=\C^n\otimes\C^k$, with $n,k\geq 1$.
In this case states and density matrices in $\A$ are in bijection,
and Prop.~\ref{prop:drho} gives the distance between arbitrary states.

Assume further that $[D,a]=0$ only if $a$ is proportional to the identity (``connectedness'' condition); note that this forces $k\geq 2$, otherwise $D\in\A$ commutes with $a=D$. In such a case it is well known that the distance is finite (see e.g.~Prop.~4.2 of \cite{CDMW09}), and then given by \eqref{eq:drho} for any pair of states.
\end{rem}

One should stress that, even if \eqref{eq:drho} may look simpler than \eqref{eq:spectrald}, it is not.
We just traded a sup for an inf in the definition of $L(\rho)$.

\section{Products of spectral triples}\label{sec:prodsp}
The product $(\A,\HH,D)$ of two spectral triples $(\A_1,\HH_1,D_1,\gamma_1)$ and $(\A_2,\HH_2,D_2)$, is defined as
\begin{equation}\label{eq:prod}
\A=\A_1\odot\A_2 \;,\qquad
\HH=\HH_1\otimes\HH_2 \;,\qquad
D=D_1\otimes 1+\gamma_1\otimes D_2 \;,
\end{equation}
where $\odot$ is the algebraic tensor product.

On the set $\mc{S}(\A_1)\times\mc{S}(\A_2)$ 
of product states of $\A$, two extended metrics are defined: the spectral distance $d_D$ and the product distance $d_{D_1}\!\boxtimes d_{D_2}$. The purpose of this section is to discuss the relation between these two.

In the Euclidean space, Pythagoras equality is a criterion that can be used to decide whether two intersecting lines are orthogonal or not. This motivates the next definition.

\begin{df}\label{def:2.2}
If
$$
d_D(\varphi,\psi)=d_{D_1}\!\boxtimes d_{D_2}\,(\varphi,\psi)
$$
we will say that the states $\varphi,\psi$ satisfy \emph{Pythagoras equality}. If such equality holds for all pure states, we will say that the product is \emph{orthogonal}.
\end{df}

The main result of \cite{DM13} is that for arbitrary unital spectral triples the distances are equivalent (although not necessarily equal), and more precisely:

\begin{thm}\label{thm:2.3}
For all product states $\varphi,\psi$:
$$
d_{D_1}\!\boxtimes d_{D_2}(\varphi,\psi)
\!\stackrel{(i)}{\rule{0pt}{8pt}\leq\rule{0pt}{0pt}}\!
d_D(\varphi,\psi)
\!\stackrel{(ii)}{\rule{0pt}{8pt}\leq\rule{0pt}{0pt}}\!\!
\sqrt{2}\,d_{D_1}\!\boxtimes d_{D_2}(\varphi,\psi)
$$
\end{thm}

\noindent
The inequality (ii) is in fact valid also when the spectral triples are non-unital, and it is a consequence of the following basic property: taking a product of spectral triples doesn't increase the horizontal resp.~vertical distance. That is:

\begin{lemma}\label{lemma:2.4}
For a product of arbitrary (not necessarily unital) spectral triples:
$$
d_D(\varphi,\psi_1\otimes\varphi_2)\leq d_{D_1}(\varphi_1,\psi_1)
\;\;\quad\text{and}\quad\;\;
d_D(\psi_1\otimes\varphi_2,\psi)\leq d_{D_2}(\varphi_2,\psi_2) \;,
$$
for all $\varphi=\varphi_1\otimes\varphi_2$ and $\psi=\psi_1\otimes\psi_2$.
\end{lemma}
\begin{proof}
For $a\in\A^{\mathrm{sa}}$ denote $a_1:=(\id\otimes\varphi_2)(a) \in\A_1^{\mathrm{sa}}$ and observe that
\begin{equation}\label{eq:spoilmarg}
\varphi(a)-(\psi_1\otimes\varphi_2)(a)=\varphi_1(a_1)-\psi_1(a_1) \;.
\end{equation}
From Lemma 9 and Corollary 11 of \cite{DM13} (I will not repeat the proof here):
$$
\|[D_1,a_1]\| \leq \|[D,a]\| \;.
$$
(There is a typo in Eq.~(18) of \cite{DM13}: there is no square in the last norm.)

Thus
\begin{align*}
d_D(\varphi,\psi_1\otimes\varphi_2)
&=\sup_{a\in\A^{\mathrm{sa}}}\big\{\varphi_1(a_1)-\psi_1(a_1):\|[D,a]\|\leq 1\big\} \\
&\leq\sup_{a\in\A^{\mathrm{sa}}}\big\{\varphi_1(a_1)-\psi_1(a_1):\|[D_1,a_1]\|\leq 1\big\}
\intertext{and this can be majorated by considering all $a_1\in\A_1^{\mathrm{sa}}$ rather than only those in the image of the map $\A^{\mathrm{sa}}\to\A_1^{\mathrm{sa}}$ above, and we get}
&\leq\sup_{a_1\in\A_1^{\mathrm{sa}}}\big\{\varphi_1(a_1)-\psi_1(a_1):\|[D_1,a_1]\|\leq 1\big\} = d_{D_1}(\varphi_1,\psi_1) \;.
\end{align*}
This proves the first inequality of Lemma \ref{lemma:2.4},
the other being similar.
\end{proof}

\begin{proof}[Proof of Theorem \ref{thm:2.3}\textnormal{(ii)}]
Now, call $a:=d_D(\varphi,\psi)$, $b:=d_D(\varphi,\psi_1\otimes\varphi_2)$ and $c:=d_D(\psi_1\otimes\varphi_2,\psi)$ the lengths of the edges of the triangle in Fig.~\ref{fig:2}. From the triangle inequality and Cauchy-Schwarz in $\R^2$:
$$
a\leq b+c=(1,1)\cdot (b,c)\leq\sqrt{2}\sqrt{b^2+c^2} \;.
$$
From this and Lemma \ref{lemma:2.4}, we get the 
inequality (ii) of Theorem \ref{thm:2.3}.
\end{proof}

In the unital case, (i) implies that the inequalities in Lemma \ref{lemma:2.4} are in fact equalities: so, the product doesn't change the horizontal resp.~vertical distance. In the non-unital case, on the other hand, there are simple counterexamples \cite[\S6]{DM13}.

\medskip

One may wonder if Theorem \ref{thm:2.3} can be generalized to arbitrary states, i.e.~if
$$
\quad
d_{D_1}\!\boxtimes d_{D_2}(\varphi^\flat,\psi^\flat)
\leq d_D(\varphi,\psi)\leq
\sqrt{2}\,d_{D_1}\!\boxtimes d_{D_2}(\varphi^\flat,\psi^\flat)
\quad ?
$$
for all $\varphi,\psi\in\mc{S}(\A)$, where
$\varphi^\flat:=\varphi_1^\flat\otimes\varphi_2^\flat$ and $\psi^\flat:=\psi_1^\flat\otimes\psi_2^\flat$ are defined using marginals.
It is easy to convince one-self that the inequality on the right can't always be true: take two different states with the same marginals, $\varphi\neq\psi$ but $\varphi^\flat=\psi^\flat$; then the product distance is zero, but $d_D(\varphi,\psi)\neq 0$, and we get a contradiction.  What fails in the proof of Lemma \ref{lemma:2.4} is the equality \eqref{eq:spoilmarg}, which holds only for product states.
On the other hand, I will show that the inequality on the left is valid for arbitrary states and unital spectral triples.

\subsection{An auxiliary extended semi-metric}\label{sec:aux}
From now on, I will assume we have a product of two \emph{unital} spectral triples (counterexamples to Theorem \ref{thm:2.3}(i) in the non-unital case can be found in \cite[\S6]{DM13}).
Consider the vector space sum:
\begin{equation}\label{eq:V}
\A_1+\A_2:=\big\{a=a_1\otimes 1+1\otimes a_2\,\big|\,a_1\in\A_1,a_2\in\A_2 \big\} \;,
\end{equation}
where we identify $\A_1$ with $\A_1\otimes 1$ and $\A_2$ with $1\otimes\A_2$.

An extended semi-metric $d^\times_D$ on $\mc{S}(\A)$ is given by
$$
d^\times_D(\varphi,\psi)=\sup_{a\in (\A_1+\A_2)^{\mathrm{sa}}}\big\{\varphi(a)-\psi(a):\|[D,a]\|\leq 1\big\}\;.
$$
Since $(\A_1+\A_2)^{\mathrm{sa}}$ is a subset of $\A^{\mathrm{sa}}$, 
we have the obvious inequality
\begin{equation}\label{eq:ineqV}
d^\times_D (\varphi,\psi)\leq d_D(\varphi,\psi) \;,
\end{equation}
for all $\varphi,\psi\in\mc{S}(\A)$.
The next one is Lemma 8 of \cite{DM13}, of which I will give a shorter proof.

\begin{lemma}
For any selfadjoint $a_1\in\A_1$ and $a_2\in\A_2$:
\begin{equation}\label{eq:lemma1}
\|[D_1,a_1]\|^2+\|[D_2,a_2]\|^2=\|[D,a_1\otimes 1+1\otimes a_2]\|^2 \;.
\end{equation}
\end{lemma}
\begin{proof}
Consider the positive operators $A:=-[D_1,a_1]^2$, $B:=-[D_2,a_2]^2$, and
note that $$C:=-[D,a_1\otimes 1+1\otimes a_2]^2=A\otimes 1+1\otimes B\;.$$
From the triangle inequality, we get $\|C\|\leq\|A\|+\|B\|$. We want to prove the opposite inequality.
We achieve this by writing the left hand side in \eqref{eq:lemma1} as:
\begin{align*}
\text{l.h.s.} &=
\sup_{\substack{v_1\in\HH_1,v_2\in\HH_2\\[2pt] \|v_1\|=\|v_2\|=1}}
\big\{ \inner{v_1,Av_1}+\inner{v_2,Bv_2} \big\}
=\sup_{\substack{v_1\in\HH_1,v_2\in\HH_2\\[2pt] \|v_1\|=\|v_2\|=1}}
\inner{v_1\otimes v_2,C(v_1\otimes v_2)}
\\
& \leq \sup_{v\in\HH:\|v\|=1}\inner{v,Cv}=\|C\|=\text{r.h.s.} \;,
\end{align*}
where last inequality comes from considering all unit vectors in $\HH$,
rather than only homogeneous tensors $v_1\otimes v_2$.
\end{proof}

\begin{prop}\label{prop:3}
For any $\varphi,\psi$ one has
$$
d_D^\times(\varphi,\psi)=d_{D_1}\!\boxtimes d_{D_2}(\varphi^\flat,\psi^\flat) \;,
$$
where, as before, $\varphi^\flat_1$ and $\varphi^\flat_2$ are the marginals of $\varphi$
and $\varphi^\flat=\varphi_1^\flat\otimes\varphi_2^\flat$.
\end{prop}
\begin{proof}
For any $a=a_1\otimes 1+1\otimes a_2\in (\A_1+\A_2)^{\mathrm{sa}}$,
$$
\varphi(a)-\psi(a)=\big\{\varphi_1^\flat(a_1)-\psi_1^\flat(a_1)\big\}+\big\{\varphi_2^\flat(a_2)-\psi_2^\flat(a_2)\big\} \;.
$$
From \eqref{eq:lemma1},
\begin{align*}
d_D^\times(\varphi,\psi) &=
\sup_{\substack{\alpha,\beta\geq 0\\[2pt] \alpha^2+\beta^2=1}}\left(
\sup_{a_1=a_1^*\in\A_1}\big\{\varphi_1^\flat(a_1)-\psi_1^\flat(a_1):\|[D,a_1]\|\leq\alpha\big\}
\right. \\[-8pt] & \hspace{3cm} \left.
+\sup_{a_2=a_2^*\in\A_2}\big\{\varphi_2^\flat(a_2)-\psi_2^\flat(a_2):\|[D,a_2]\|\leq\beta\big\}
\right) \\[3pt]
=\sup_{\substack{\alpha,\beta\geq 0\\[2pt] \alpha^2+\beta^2=1}} &
\left(\alpha\,d_{D_1}(\varphi_1^\flat,\psi_1^\flat)+\beta\,d_{D_2}(\varphi_2^\flat,\psi_2^\flat)\right)
=\sqrt{d_{D_1}(\varphi_1^\flat,\psi_1^\flat)^2+d_{D_2}(\varphi_2^\flat,\psi_2^\flat)^2} \;,
\end{align*}
where last equality comes from Cauchy-Schwarz inequality (cf.~\cite[Lemma 7]{DM13} for the details).
\end{proof}

As a corollary, from last proposition and \eqref{eq:ineqV}, we get
\begin{equation}\label{eq:ineqDtimes}
d_D(\varphi,\psi)\geq d_{D_1}\!\boxtimes d_{D_2}(\varphi^\flat,\psi^\flat) 
\end{equation}
for arbitrary states.
Clearly, since the product distance depends only on the marginals, for $\varphi\neq\psi$ with the same marginals the right hand side of \eqref{eq:ineqDtimes} is zero, while $d_D$ is not. Whatever the spectral triples are, \eqref{eq:ineqDtimes} cannot be an equality on arbitrary states, but it could be on the set of product states, where $d_D^\times$ is a proper extended metric (or on some smaller set).

\subsection{Product states}\label{sec:prodstates}
In this section we give a criterion (Prop.~\ref{prop:5}) to check whether $d_D$ is the product metric on product states. As an application, I will show in \S\ref{sec:C2} that in a product of two two-point spaces, Pythagoras equality holds for arbitrary product states.

For any $\varphi=\varphi_1\otimes\varphi_2$ and $\psi=\psi_1\otimes\psi_2$
we define a map $P_{\varphi,\psi}:\A\to\A$ by:
\begin{equation}\label{eq:Pvp}
P_{\varphi,\psi}:=\varphi_1^\sharp\otimes\id+\id\otimes\psi_2^\sharp-\varphi_1^\sharp\otimes\psi_2^\sharp \;.
\end{equation}
One can verify that $\id-P_{\varphi,\psi}$
(and then $P_{\varphi,\psi}$) is an idempotent by looking at the identity:
\begin{equation}\label{eq:onemP}
\id-P_{\varphi,\psi}=(\id-\varphi_1^\sharp)\otimes(\id-\psi_2^\sharp)
\end{equation}
and noting that $\varphi_1^\sharp,\psi_2^\sharp$ (and then $\id-\varphi_1^\sharp$ and $\id-\psi_2^\sharp$) are idempotents.

\begin{lemma}
We can decompose $\A$ as a direct sum (of vector spaces, not algebras):
\begin{equation}\label{eq:sum}
\A=\ran(P_{\varphi,\psi})\oplus\ker(P_{\varphi,\psi}) \;,
\end{equation}
where range and kernel of the idempotent are:
$$
\ran(P_{\varphi,\psi})=\A_1+\A_2 \;,\qquad
\ker(P_{\varphi,\psi})=\ker\varphi_1\otimes\ker\psi_2 \;,
$$
and $\A_1+\A_2$ is the set \eqref{eq:V}.
\end{lemma}

\begin{proof}
Clearly $\ran(P_{\varphi,\psi})\subset\A_1+\A_2$.
But a simple check proves that $P_{\varphi,\psi}$ is the identity on $\A_1+\A_2$, hence $a=P_{\varphi,\psi}(a)\;\forall\;a\in\A_1+\A_2$ and the opposite inclusion $\A_1+\A_2\subset\ran(P_{\varphi,\psi})$ holds too, proving that the two sets coincide.
The range of $\id-\varphi_1^\sharp$ is the kernel of $\varphi_1$, and similar for $\id-\psi_2^\sharp$.
Hence the range of $\id-P_{\varphi,\psi}$, which is the kernel of $P_{\varphi,\psi}$, is
$\ker\varphi_1\otimes\ker\psi_2$.
Any $a\in\A$ can be written in a unique way as \mbox{$a=P_{\varphi,\psi}(a)+(\id-P_{\varphi,\psi})(a)$},
proving \eqref{eq:sum}.
\end{proof}

\begin{prop}\label{prop:5}
Let $\varphi=\varphi_1\otimes\varphi_2$ and $\psi=\psi_1\otimes\psi_2$. If,
for all $a\in\A^{\mathrm{sa}}$,
\begin{equation}\label{eq:ineqD}
\|[D,P_{\varphi,\psi}(a)]\|\leq\|[D,a]\| \;,
\end{equation}
then $d_D(\varphi,\psi)=d_D^\times(\varphi,\psi)$ (so, the states $\varphi,\psi$ satisfy Pythagoras equality).
\end{prop}

\begin{proof}
From $\ran(\id-P_{\varphi,\psi})=\ker\varphi_1\otimes\ker\psi_2$, 
$\id-P_{\varphi,\psi}$ maps $\A$ into the kernel of both $\varphi$ and $\psi$.
So, for all $a\in\A$:
$$
\varphi(P_{\varphi,\psi}(a))=\varphi(a) \;,\qquad
\psi(P_{\varphi,\psi}(a))=\psi(a) \;.
$$
Moreover, $b=P_{\varphi,\psi}(b)$ for all $b\in\A_1+\A_2$. Thus:
\begin{align*}
d_D^\times(\varphi,\psi)
&=\sup_{b\in (\A_1+\A_2)^{\mathrm{sa}}}\big\{\varphi(P_{\varphi,\psi}(b))-\psi(P_{\varphi,\psi}(b)):\|[D,P_{\varphi,\psi}(b)]\|\leq 1\big\}
\\
&=\sup_{a\in\A^{\mathrm{sa}}}\big\{\varphi(P_{\varphi,\psi}(a))-\psi(P_{\varphi,\psi}(a)):\|[D,P_{\varphi,\psi}(a)]\|\leq 1\big\}
\\
&=\sup_{a\in\A^{\mathrm{sa}}}\big\{\varphi(a)-\psi(a):\|[D,P_{\varphi,\psi}(a)]\|\leq 1\big\} \;.
\end{align*}
Now, from \eqref{eq:ineqD}, we deduce that $d_D^\times(\varphi,\psi)\geq d_D(\varphi,\psi)$.
But the opposite inequality also holds, cf.\ \eqref{eq:ineqV}, hence the thesis.
\end{proof}

\subsection{Normal states}
Let $\varphi=\varphi_1\otimes\varphi_2$ and $\psi=\psi_1\otimes\psi_2$ be as in previous section.
In this section, I assume that $\varphi_1$ and $\psi_2$ are normal states with density matrices $\rho_1,\rho_2$:
$$
\varphi_1(a)=\tr_{\HH_1}(\rho_1a) \;,\qquad
\psi_2(b)=\tr_{\HH_2}(\rho_2b) \;,
$$
for all $a\in\A_1$ and $b\in\A_2$. The above formulas allow then to extend
$\varphi_1$ and $\psi_2$ to the whole $\B(\HH_1)$ resp.~$\B(\HH_2)$,
and $P_{\varphi,\psi}$ to the whole $\B(\HH)$. Note however that,
since in general the density matrix of a given state is not unique, the extension is also
not unique.

Let $K$ be the norm of the idempotent \eqref{eq:Pvp}:
$$
K:=\sup_{b\in\B(\HH):b\neq 0}\frac{\|P_{\varphi,\psi}(b)\|}{\|b\|}
$$
\begin{lemma}\label{lemma:K3}
$1\leq K\leq 3$.
\end{lemma}
\begin{proof}
Any $b\in\ran(P_{\varphi,\psi})$ gives a lower bound $1$.
The upper bound comes from the norm $1$ property of a state together with the triangle inequality.
\end{proof}

\begin{lemma}
Assume $\rho_1$ commutes with $\gamma_1$ and either: i) $\rho_2$ commutes with $D_2$ (it can be, for example, an eigenstate of the
Dirac operator); or ii) $(\A_2,\HH_2,D_2,\gamma_2)$ is also even (although we do not use $\gamma_2$ in the definition of the product
spectral triple) and $\rho_2$ commutes with $\gamma_2$. Then:
\begin{equation}\label{eq:nullA}
\varphi_1(\gamma_1[D_1,a])=0 \;,\qquad\forall\;a\in\A_1,
\end{equation}
and
\begin{equation}\label{eq:nullB}
\psi_2([D_2,b])=0 \;,\qquad\forall\;b\in\A_2.
\end{equation}
\end{lemma}
\begin{proof}
If $[\rho_1,\gamma_1]=0$, since $\gamma_1$ anticommutes with 
$D_1$ and commutes with $\A_1$, then for all $a\in\A_1$:
\begin{align*}
\varphi_1(\gamma_1[D_1,a]) &=-\varphi_1([D_1,a]\gamma_1)
=-\tr_{\HH_1}(\rho_1[D_1,a]\gamma_1)
\\
\intertext{and from the cyclic property of the trace:}
&=-\tr_{\HH_1}(\gamma_1\rho_1[D_1,a])
=-\tr_{\HH_1}(\rho_1\gamma_1[D_1,a])
\\
&=-\varphi_1(\gamma_1[D_1,a]) \;,
\end{align*}
This implies \eqref{eq:nullA}.

The proof of \eqref{eq:nullB} under the hypothesis (ii) is similar.
If $[\rho_2,\gamma_2]=0$, then
\begin{align*}
\psi_2([D_2,b]) &=
\psi_2(\gamma_2^2[D_2,b])=
-\psi_2(\gamma_2[D_2,b]\gamma_2) \\
&=-\tr_{\HH_2}(\rho_2\gamma_2[D_2,b]\gamma_2)
=-\tr_{\HH_2}(\gamma_2\rho_2\gamma_2[D_2,b]) \\
&=-\tr_{\HH_2}(\rho_2\gamma_2^2[D_2,b])=-\psi_2([D_2,b]) \;,
\end{align*}
hence $\psi_2([D_2,b])=0$.

If, on the other hand, (i) is satisfied, that is $[D_2,\rho_2]=0$, then:
$$
\psi_2([D_2,b])=\tr_{\HH_2}([D_2,\rho_2b])=0 \;.
$$
(Since $D_2$ is unbounded, the cyclic property of the trace doesn't hold.
But one can prove the above equality by writing the trace in an eigenbasis
of $D_2$ and $\rho_2$.)
\end{proof}

\begin{prop}\label{prop:9}
Let $\varphi_1$ and $\psi_2$ be normal and satisfying \eqref{eq:nullA} and
\eqref{eq:nullB}. Then
$$
d_D^\times(\varphi,\psi) \leq d_D(\varphi,\psi)\leq K\,d_D^\times(\varphi,\psi) \;.
$$
If $K=1$, then Pythagoras equality holds.
\end{prop}

\begin{proof}
Let $\gamma:=\gamma_1\otimes 1$ and
$$
\widetilde P_{\varphi,\psi}(b)=\gamma P_{\varphi,\psi}(\gamma b)
$$
for all $b\in\B(\HH)$. For any $y=\gamma x$, $y^*y=x^*x$ and then $\|y\|=\|x\|$. Hence
$$
\|\widetilde P_{\varphi,\psi}\|=\sup_{b\in\B(\HH):b\neq 0}\frac{\|\gamma P_{\varphi,\psi}(b')\|}{\|\gamma b'\|}=
\|P_{\varphi,\psi}\|=K \;,
$$
where $b'=\gamma b$.
For all $a\in\A$:
\begin{align*}
[D,P_{\varphi,\psi}(a)]
&=[D_1\otimes 1,\id\otimes\psi_2^\sharp(a)]+[\gamma_1\otimes D_2,\varphi_1^\sharp\otimes\id(a)]
\\
&=(\id\otimes\psi_2^\sharp)([D_1\otimes 1,a])+(\gamma_1\otimes 1)(\varphi_1^\sharp\otimes\id)([1\otimes D_2,a])
\\
&=(\id\otimes\psi_2^\sharp)([D_1\otimes 1,a])+(\widetilde{\varphi}_1^{\,\sharp}\otimes\id)([\gamma_1\otimes D_2,a])
\\
&=\widetilde P_{\varphi,\psi}([D_1\otimes 1,a])+\widetilde P_{\varphi,\psi}([\gamma_1\otimes D_2,a])
\\
&=\widetilde P_{\varphi,\psi}([D,a]) \;.
\end{align*}
In last-but-one equality we used \eqref{eq:nullA} and \eqref{eq:nullB}.
Thus $\|[D,P_{\varphi,\psi}(a)]\|\leq K\|[D,a]\|$.
If we now repeat the proof of
Prop.~\ref{prop:5}, we deduce that $d_D^\times(\varphi,\psi)\geq K^{-1}d_D(\varphi,\psi)$,
hence the thesis.
\end{proof}

While Prop.~\ref{prop:5} can be used to prove Pythagoras equality in some examples (cf.~\S\ref{sec:C2}), it is not clear if there are examples where $K\leq\sqrt{2}$, and then Prop.~\ref{prop:9} allows to improve the bounds in Theorem \ref{thm:2.3}. Note that a non-trivial $C^*$-algebra projection has always norm $K=1$\footnote{$P=P^*P$ implies $\|P\|^2=\|P^*P\|=\|P\|$ by the $C^*$-identity, hence $\|P\|$ is $0$ or $1$.}. Unfortunately this is not true for idempotent endomorphisms of a normed vector space. For example, if $P\in M_n(\C)$ is the matrix with all $1$'s in the first row and zero everywhere else, then $P^2=P$, but $PP^*=\mathrm{diag}(n,0,\ldots,0)$ implies that the norm is $\sqrt{n}$.

\section{Examples from classical and quantum transport}\label{sec:4}

In order to understand what the spectral distance looks like, it is useful to have in mind some examples.
In this section, I collect some examples where the distance can be explicitly computed, which include: the canonical spectral triple of a finite metric space and of a Riemannian manifold, which are useful to illustrate the connection with transport theory, several natural spectral triples for the state space of a ``qubit'' (\S\ref{sec:Bloch}), and a digression on the Wasserstein distance between quantum states for the Berezin quantization of a homogeneous space.

\subsection{The two-point space}\label{sec:twopoints}
Let us start with the simplest example, $\A=\C^2$.
We identify $\A$ with the subalgebra of $M_2(\C)$ of diagonal matrices,
acting on the Hilbert space $\HH=\C^2$ via matrix multiplication.
We get an even spectral triple $(\A,\HH,D,\gamma)$ by taking $D:=\frac{1}{2\lambda}F$, where:
\begin{equation}\label{eq:gammaF}
\gamma:=\left[\!\begin{array}{rr} 1 & 0 \\ 0 & \!-1 \end{array}\!\right] \;,\qquad\quad
F:=\left[\!\begin{array}{rr} 0 & \;1 \\ 1 & 0 \end{array}\!\right] \;.
\end{equation}
Here $\lambda>0$ is a fixed length parameter.
We can identify $\varphi\in[-\lambda,\lambda]$ with the state:
\begin{equation}\label{eq:stateC2}
a=\left[\!\begin{array}{cc} a_\uparrow & 0 \\ 0 & a_\downarrow \end{array}\!\right]
\mapsto \frac{1+\varphi/\lambda}{2}\,a_\uparrow+\frac{1-\varphi/\lambda}{2}\,a_\downarrow \;.
\end{equation}
Clearly $\mc{S}(\A)\simeq [-\lambda,\lambda]$,
and I will use the same symbol to denote a state and the corresponding point in $[-\lambda,\lambda]$.
Pure states are given by $\varphi=\pm\lambda$.

The proof that $d_D(\varphi,\psi)=|\varphi-\psi|$ is an exercise that I leave to the reader.

\subsection{The standard \texorpdfstring{$2$}{2}-simplex}\label{sec:equi}
After the two-point space, the simplest example is the space with three points at equal distance.
Let $\A=\C^3$ (with componentwise multiplication), represented on $\HH=\A\otimes\A$ by diagonal multiplication,
and $D$ be given by $D(a\oplus b)=D_-(b)\oplus D_+(a)$ where $D_+$ is the permutation matrix
$$
D_+=\left[\begin{array}{ccc}
0 & 0 & 1 \\
1 & 0 & 0 \\
0 & 1 & 0
\end{array}\right]
$$
and $D_-=(D_+)^*$.
States are in bijection with points of the simplex (Fig.~\ref{fig:3a}),
$$
\Delta^2:=\big\{\varphi\in\R^3\,:\,\varphi_i\geq 0\;\forall\;i\;\text{and}\;\varphi_1+\varphi_2+\varphi_3=1 \big\} \;,
$$
via the map
$$
\varphi=(\varphi_1,\varphi_2,\varphi_3)\mapsto\varphi^\star:\varphi^\star(a):=\inner{\varphi,a} \;,\qquad\forall\;a\in\A,
$$
where on the right we have the canonical inner product of $\C^3$. With a slight abuse of notation, from now on I will denote by $\varphi$ (without $\star$) both the vector and the corresponding state of $\A$.
The vertices $e_1=(1,0,0)$, $e_2=(0,1,0)$ and $e_3=(0,0,1)$ are mapped to pure states, and the barycenter
$$
e_0:=\tfrac{1}{3}(e_1+e_2+e_3)
$$
of $\Delta^2$ to the trace state.
For all $\varphi,\psi\in\Delta^2$, the difference $\varphi-\psi$ belongs to the plane through the origin (parallel to $\Delta^2$ and orthogonal to $e_0$) given by:
$$
\Pi:=\big\{v\in\R^3:v_1+v_2+v_3=0\big\} \;.
$$

\begin{figure}[t]
  \centering
  \subfloat[\label{fig:3a}The standard $2$-simplex]
	{\makebox[5cm]{\includegraphics[height=5cm]{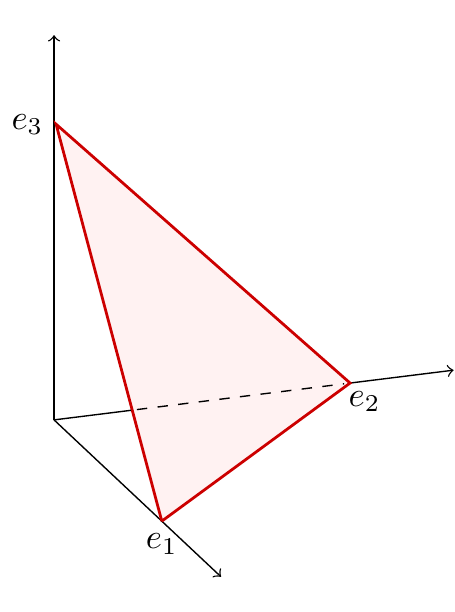}}}
  \subfloat[\label{fig:3b}The triangle $T$ (red), $Q(\Delta^2)$ (black), and the hexagon $E$.]
  {\makebox[8cm]{\includegraphics[height=5cm]{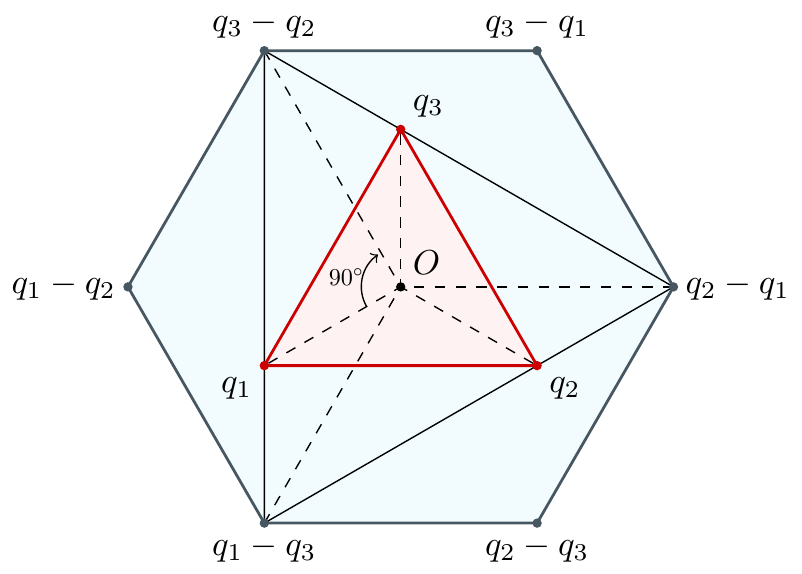}}}

	\caption{}
\end{figure}

\begin{lemma}\label{lemma:DaC3}
$\|[D,a]\|=\max_{i\neq j}|a_i-a_j|$ for all $a\in\A^{\mathrm{sa}}$.
\end{lemma}
\begin{proof}
$[D,a]^*[D,a]$ is the direct sum of the operators $[D_+,a]^*[D_+,a]$ and $[D_-,a]^*[D_-,a]=[D_+,a][D_+,a]^*$, which by the $C^*$-identity have the same norm. So $\|[D,a]\|=\|[D_+,a]\|$. From
$$
[D_+,a]^*[D_+,a]=\left[\begin{array}{ccc}
(a_1-a_2)^2 & 0 & 0 \\
0 & (a_2-a_3)^2 & 0 \\
0 & 0 & (a_1-a_3)^2
\end{array}\right]
$$
we get the thesis.
\end{proof}

With this, one can prove that $d_D$ is the metric $d_D(e_i,e_j)=1\;\forall\;i\neq j$ on pure states, and the Chebyshev metric on arbitrary states, as I show below.

With the vector product we define a surjective linear map $Q:\R^3\to\Pi$,
$$
Q(a):=a\wedge 3e_0=(a_2-a_3,a_3-a_1,a_1-a_2) \;,
$$
satisfying the algebraic identity:
\begin{equation}\label{eq:Qiso}
\inner{Q(v),Q(a)}=3\inner{v,a} \;,\qquad\forall\;v\in\Pi,a\in\R^3\;.
\end{equation}
The translation $a\to a-e_0$, orthogonal to $\Delta^2$, transforms $\Delta^2$ into a new equilateral triangle $T$ with barycenter at the origin $O$, and with vertices $\{q_i=e_i-e_0\}_{i=1}^3$.

For $a\in\Delta^2$, since $Q(a)=Q(a-e_0)$, we can imagine that $Q$ is the composition of the translation $\Delta^2\to T$ and the endomorphism $Q|_{\Pi}$ of the plane $\Pi\supset T$.
Since $v\in\Pi$ is orthogonal to $e_0$, $Q(v)=v\wedge 3e_0$ is a $90^\circ$ clockwise rotation composed with a dilatation (uniform scaling) by a factor $\|3e_0\|_2=\sqrt{3}$.
The situation is the one illustrated in Fig.~\ref{fig:3b}: the small triangle $T$ is rotated and then scaled until it matches the big circumscribed triangle $Q(\Delta^2)$,
with vertices $\{Q(e_i)\}_{i=1}^3$ and triple area.\footnote{It's easy to check that $T$ is inscribed in $Q(\Delta^2)$: for example, from the algebraic identity $q_3=\frac{2}{3}Q(e_1)+\frac{1}{3}Q(e_3)$ we see that $q_3$ is a convex combination of two vertices of $Q(\Delta^2)$, and then lie on the corresponding edge.}

\medskip

Let us come back to the spectral distance. The usefulness of the map $Q$ is that
\begin{equation}\label{eq:Dsup}
\|[D,a]\|=\|Q(a)\|_\infty\;,\qquad\forall\;a\in\A^{\mathrm{sa}},
\end{equation}
cf.~Lemma \ref{lemma:DaC3}.

\begin{lemma}
For all $\varphi,\psi\in\Delta^2$:
\begin{equation}\label{eq:auxD}
d_D(\varphi,\psi)=
\frac{1}{3}\sup_{w\in\Pi}\big\{ \inner{Q(\varphi-\psi),w} : \|w\|_\infty\leq 1\big\} \;.
\end{equation}
\end{lemma}
\begin{proof}
For any $a\in\A^{\mathrm{sa}}=\R^3$, called $v=\varphi-\psi$, then
$\varphi(a)-\psi(a)=\inner{v,a}=\frac{1}{3}\inner{Q(v),Q(a)}$.
From this and \eqref{eq:Dsup}, we get \eqref{eq:auxD} by 
noticing that the map $\R^3\to\Pi$, $a\mapsto w=Q(a)$ is surjective.
\end{proof}

The distance can now be explicitly computed.

\begin{prop}\label{prop:triangle}
For all $\varphi,\psi\in\Delta^2$:
$$
d_D(\varphi,\psi)=\|\varphi-\psi\|_\infty \;.
$$
\end{prop}

\begin{proof}
The set
$$
E:=\{w\in\Pi:\|w\|_\infty\leq 1\}
$$
which appears in \eqref{eq:auxD} is the hexagon
(in the plane $\Pi$) with vertices $\{q_i-q_j\}_{i\neq j}$, cf.~Figure \ref{fig:3b}.
We want to maximize the product
$$
\inner{Q(\varphi-\psi),w}=\|Q(\varphi-\psi)\|_2\|w'\|_2
=\sqrt{3}\,\|\varphi-\psi\|_2\|w'\|_2
 \;,
$$
where $w\in E$ and $w'$ is the orthogonal projection in the direction of $Q(\varphi-\psi)$.
Clearly the maximum is reached when $w$ is in the boundary of $E$, and more precisely
when $w$ is the nearest vertex, cf.~Figure \ref{fig:4a}. So, the maximum
value of $\|w'\|$ is $\sqrt{2}\cos\alpha$, and
$$
d_D(\varphi,\psi)=
\sqrt{\frac{2}{3}}\,\|\varphi-\psi\|_2\cos\alpha \;.
$$
The angle between $Q(\varphi-\psi)$ and the nearest vertex of the hexagon $E$ is the same
as the angle between $\varphi-\psi$ and the nearest vertex of the rotated hexagon $\frac{1}{3}Q(E)$.
Since $\frac{1}{3}Q(E)$ has vertices $\{\pm q_i\}_{i=1,2,3}$, then (minimal angle means maximum
cosine):
$$
\cos\alpha=\max_{i=1,2,3}\frac{\left|\inner{\varphi-\psi,q_i}\right|}{\|\varphi-\psi\|_2\|q_i\|_2} \;.
$$
Since $\|q_i\|_2=\sqrt{2/3}$, we get
\begin{equation}\label{eq:tempdD}
d_D(\varphi,\psi)=\max_{i=1,2,3}\left|\inner{\varphi-\psi,e_i-e_0}\right| \;,
\end{equation}
but $\inner{\varphi-\psi,e_0}=0$, so
$d_D(\varphi,\psi)=
\max_{i=1,2,3}|\varphi_i-\psi_i|=
\|\varphi-\psi\|_{\infty}$.
\end{proof}

\begin{figure}[t]
  \hspace*{-1cm}%
  \subfloat[\label{fig:4a}Computing \eqref{eq:auxD}.\hspace*{8mm}]{\makebox[8cm]{\hspace*{8mm}\includegraphics[height=4cm]{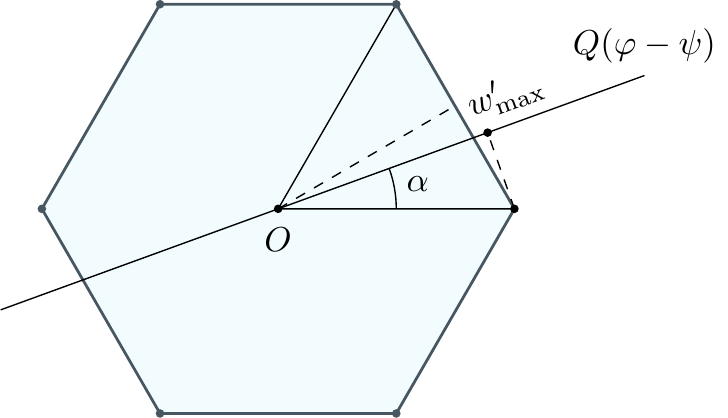}}}
  \subfloat[\label{fig:4b}The Reuleaux triangle.\hspace*{-10pt}]
  {\includegraphics[height=4cm]{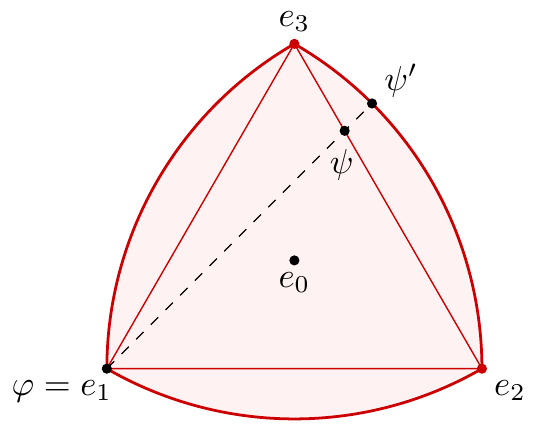}}

	\caption{}
\end{figure}

From \eqref{eq:tempdD} we get a nice geometrical interpretation.
The distance is obtained by projecting the vector $\varphi-\psi$
on the three medians of the simplex, and computing the max of the lengths of the three resulting segments.

\smallskip

To get an idea of the behavior of $d_D$, we can compute its value when one state is a vertex, and the other is on the opposite edge of the simplex, say e.g.~$\varphi=e_1$ and $\psi=(1-\lambda) e_2+\lambda e_3$ for some $\lambda\in [0,1]$. Then
$d_D(\varphi,\psi)=\max\{1,\lambda,1-\lambda\}=1$ is independent of $\lambda$.

If we draw a picture, the simplex with the spectral distance
looks like a Reuleaux triangle, cf.~Figure \ref{fig:4b}. In the figure, $\psi'$ is the intersection of the line through $\varphi$ and $\psi$ with the circle centered at $\varphi$ and passing through the opposite vertices; $d_D(\varphi,\psi)$ is the Euclidean distance between $\varphi$ and $\psi'$.

\subsection{Finite metric spaces}\label{sec:finite}
Given an arbitrary finite metric space $(X,g)$, there is a canonical even spectral triple $(\A,\HH,D,\gamma)$ associated to it, with the property that the spectral distance between pure states coincides with the original metric $g$.
This is constructed as follows.

We take $\A=\C^N\simeq C(X)$ and think of the pure states
\begin{equation}\label{eq:xxstar}
x_i(a):=a_i \;,\qquad\forall\;a=(a_1,\ldots,a_N)\in\C^N,\;i=1,\ldots,N,
\end{equation}
as the $N$-points of $X$.
For $i,j=1,\ldots,N$, let $\HH_{ij}=\C^2$ and $\pi_{ij}$ be the representation
$$
\pi_{ij}(a)=\left[\!\begin{array}{cc} a_i & 0 \\ 0 & a_j \end{array}\!\right] \;.
$$
A faithful (unital $*$-)representation of $\A$ is given on $\HH=\bigoplus_{i\neq j}\HH_{ij}$ by $\pi=\bigoplus_{i\neq j}\pi_{ij}$. I will identify $\A$ with $\pi(\A)$ and omit the symbol $\pi$.
Finally, we define
$D=\bigoplus_{i\neq j}D_{ij}$
and
$\gamma=\bigoplus_{i\neq j}\gamma_{ij}$,
where $D_{ij}$ and $\gamma_{ij}$ are the operators on $\HH_{ij}$:
$$
D_{ij}=\frac{1}{g_{ij}}\left[\begin{array}{rr} 0 & 1 \\ 1 & 0 \end{array}\right]
\;,\qquad
\gamma_{ij}=\left[\begin{array}{rr} 1 & 0 \\ 0 & \!\!-1 \end{array}\right] \;,
$$
and I use the shorthand notation
$g_{ij}:=g(x_i,x_j)$ (note that $g_{ij}\neq 0\;\forall\;i\neq j$).

It is not difficult to prove that the spectral distance between pure states is the metric we started from.

\begin{prop}\label{prop:3.1}
$d_D(x_k,x_l)=g_{kl}$ for all $k,l=1,\ldots,N$.
\end{prop}
\begin{proof}
Since $[D,a]=\bigoplus_{i\neq j}[D_{ij},\pi_{ij}(a)]$, clearly
\begin{equation}\label{eq:clearly}
\|[D,a]\|=
\max_{i\neq j}\|[D_{ij},\pi_{ij}(a)]\|=
\max_{i\neq j}\left|\frac{a_i-a_j}{g_{ij}}\right| \;,
\end{equation}
and then
$$
d_D(x_k,x_l)=\sup_{a\in\R^N}\big\{a_k-a_l:\|[D,a]\|\leq 1\big\}\leq g_{kl} \;.
$$
The inequality is saturated by the element $\bar a=(g_{1l},g_{2l},\ldots,g_{Nl})\in\R^N$, which satisfies
$$
\left|\frac{\bar a_i-\bar a_j}{g_{ij}}\right|
\leq \frac{|\bar a_i|+|\bar a_j|}{g_{ij}}
=\frac{g_{il}+g_{jl}}{g_{ij}}
\leq 1
$$
due to the triangle inequality --- so $\|[D,\bar a]\|\leq 1$ --- and
$$
\bar a_k-\bar a_l=g_{kl}-g_{ll}=g_{kl} \;. \vspace{-12pt}
$$
\end{proof}

This is basically the construction of \cite[\S9]{IKM01} (see also
\cite[\S2.2]{SuijBook}). This spectral triple is the discrete analogue of the Hodge spectral triple discussed in \S\ref{sec:Hodge}.

Note that the construction, as well as Prop.~\ref{prop:3.1}, remains valid if $g$ is an extended metric (with the convention that in the definition of $D$,
$D_{ij}=0$ if $g_{ij}=\infty$), but doesn't extend to semi-metrics
(we need $g_{ij}\neq 0\;\forall\;i\neq j$ in order to define $D$).
Since $d_D$ is an extended metric, it is not surprising that $g$ must be an extended metric in order for the proof of Prop.~\ref{prop:3.1} to work.

\medskip

It is easy to give a transport theory interpretation to $d_D(\varphi,\psi)$ in the finite case. A general state has the form
$\varphi=\sum\nolimits_i\varphi_ix_i$,
where $\{\varphi_i\}_{i=1}^N$ is a probability distributions on $X$.

As usual, suppose $\varphi$ describes the distribution of some material
in $X$, which we want to move to resemble another distribution $\psi$.
A \emph{transport plan} will be then described by a matrix $P=(p_{ij})\in M_N(\R)$ with non-negative entries: $p_{ij}\geq 0$. The quantity $p_{ij}$ tells us the \emph{fraction} of material we move from $x_i$ to $x_j$ (with $p_{ii}$ the fraction which remains in the site $i$), which means
\begin{equation}\label{eq:linA}
\sum\nolimits_kp_{ik}=1 \qquad\quad\forall\;i=1,\ldots,N.
\end{equation}
A matrix $P$ satisfying the above conditions is called \emph{stochastic}.

After the movement, the total material we want at the site $x_j$ is $\psi_j$.
Since at the beginning the amount was $\varphi_j$, the difference between what we move in and what we move out from $x_j$ must be $\psi_j-\varphi_j$.
The transport plan must then satisfy the condition
\begin{equation}\label{eq:inout}
\text{(in)}-\text{(out)}=\sum\nolimits_k\varphi_kp_{kj}
-\sum\nolimits_k\varphi_jp_{jk}
=\psi_j-\varphi_j \qquad\quad\forall\;j=1,\ldots,N,
\end{equation}
which due to \eqref{eq:linA} is equivalent to:
\begin{equation}\label{eq:linB}
\sum\nolimits_k\varphi_kp_{kj}=\psi_j \qquad\quad\forall\;j=1,\ldots,N.
\end{equation}
If $c_{ij}$ is the unit cost to move from $x_i$ to $x_j$, the minimum cost for a transport $\varphi\to\psi$ will be:
\begin{equation}\label{eq:cost}
K(\varphi,\psi):=\inf\;\sum\nolimits_{ij}\varphi_ic_{ij}p_{ij}
\end{equation}
where the inf is over all transport plans $P$ from $\varphi$ to $\psi$.

The set of transport plans $\varphi\to\psi$ is a subset of $[0,1]^{N^2}\subset\R^{N^2}$ defined by the system of linear equations \eqref{eq:linA} and \eqref{eq:linB}. It is the intersection of an affine subspace of 
$\R^{N^2}$ with a hypercube: so it's closed and bounded, hence compact (by Heine-Borel theorem). Since $P\mapsto \sum\nolimits_{ij}\varphi_ic_{ij}p_{ij}$ is a linear (hence continuous) function on a compact space,
by Weierstrass theorem the infimum in \eqref{eq:cost} is actually a minimum
(see e.g.~Appendix E of \cite{HJ90}), thus justifying the terminology.

Next proposition is the discrete (finite) version of Kantorovich duality, which can be easily derived from Exercise 1.7 of \cite{Vil03}.
Next lemma can be found in the appendix of \cite{Ev97}, and we repeat the proof here for the reader's ease.

\begin{lemma}
For all $\varphi,\psi$:
\begin{equation}\label{eq:kant}
K(\varphi,\psi)=\sup_{a,b\in\R^N}\big\{
\varphi(a)+\psi(b)\,:\,a_i+b_j\leq c_{ij}\;\forall\;i,j
\big\} \;.
\end{equation}
\end{lemma}

\begin{proof}
Since every finite-dimensional real inner product space is isometric to $\R^n$ for some $n$, we can restate the duality relation in \cite[Ex.~1.7]{Vil03} as follows. Let $V,W$ be two finite-dimensional real inner product spaces. Then for any $v\in V$, $w\in W$ and $f:V\to W$ we have
\begin{equation}\label{eq:minmax}
\sup_{u\in V}\big\{\inner{v,u}_V\,:\,w-f(u)\succeq 0\big\}=
\inf_{p\in W}\big\{ \inner{w,p}_W\,:\,p\succeq 0\;\text{and}\;f^T(p)=v
\big\} \;,
\end{equation}
where $\succeq 0$ means that all components of the vector are non-negative, and the transpose $f^T$ of $f$ is the map \mbox{$f^T:W\to V$} defined in the usual way:
$$
\inner{u,f^T(h)}_V=\inner{f(u),h}_W \;,\qquad\forall\;u\in V,h\in W.
$$
Now we apply this to $V=\R^{2N}$ with canonical inner product, and $W=M_N(\R)$ with weighted Hilbert-Schmidt inner product:
$$
\inner{e,h}_W=\sum\nolimits_{i,j}\varphi_i\hspace{1pt}e_{ij}\hspace{1pt}h_{ij} \;,\qquad\forall\;e,h\in W.
$$
Take
$$
v=(\varphi_1,\ldots,\varphi_N,\psi_1,\ldots,\psi_N) \;,\qquad
w=(g_{ij}) \;,
$$
and $f$ the function
$$
f(u)_{ij}=a_i+b_j\;,\qquad\forall\;u=(a_1,\ldots,a_N,b_1,\ldots,b_N)\in V.
$$
The transpose is
$$
f^T(h)_k=\varphi_k\sum\nolimits_jh_{kj} \;,\qquad
f^T(h)_{k+N}=\sum\nolimits_i\varphi_ih_{ik} \;,
$$
for all $k=1,\ldots,N$ and $h\in W$.

The condition $f^T(p)=v$ is equivalent to \eqref{eq:linB} plus
the condition $\sum\nolimits_kp_{ik}=1$ for all $i$ such that $\varphi_i\neq 0$.
On the other hand if $\varphi_{i_0}=0$ for some $i_0$, the row $i_0$ of $P$ does not contribute to $\inner{g,p}_W$, and we can always replace the row $i_0$ of $P$ by an arbitrary probability distribution and find a new $P$ satisfying \eqref{eq:linA} too. The right hand side of \eqref{eq:minmax} is then equal to $K(\varphi,\psi)$.
Computing the left hand side of \eqref{eq:minmax} one then gets \eqref{eq:kant}.
\end{proof}

\begin{prop}
If the unit cost is $c_{ij}=g_{ij}$, then $d_D(\varphi,\psi)=K(\varphi,\psi)\;\forall\;\varphi,\psi$.
\end{prop}

\begin{proof}
One can follow the proof of Theorem 1.14 in \cite{Vil03} (there is basically no simplification in the finite case), and prove that $K(\varphi,\psi)$ is the supremum of $\varphi(a)-\psi(a)$ over $a\in\A$ satisfying $|a_i-a_j|\leq g_{ij}$, which by \eqref{eq:clearly} is equivalent to $\|[D,a]\|\leq 1$.
\end{proof}

An analogue of the Wasserstein distance for quantum states is discussed in \S7.7 of \cite{BZ06}. The advantage of a distance which can be defined both as an inf and a sup, is that from the former definition one can get upper bounds, and from the latter one gets lower bounds. With some luck, these will coincide and allow to explicitly compute the distance.
What is missing for the spectral distance is a formulation as an infimum.

In \cite{ZS01}, the authors define a ``Monge distance'' between quantum states as the Wasserstein distance between the corresponding Husimi distributions of quantum optics.
They also rewrite such a distance as a sup, dual to a seminorm on an operator space in the spirit of Rieffel's quantum metric spaces \cite{Rie99,Rie03} (cf.~Prop.~4 of \cite{ZS01}). It is not clear however if such a seminorm comes from a Dirac operator.

Following the approach of \cite{ZS01}, I will explain in \S\ref{sec:qWass} how to define a ``Wasserstein'' distance between quantum states using the symbol map of Berezin quantization.

\subsection{Riemannian manifolds}\label{sec:Riem}

Let $M$ be a finite-dimensional oriented Riemannian spin manifold with no boundary. A canonical even spectral triple can be defined as follows
\begin{equation}\label{eq:HD}
  \A=C^\infty_0(M), \qquad \HH=\Omega^\bullet(M), \qquad
  D=\de+\de^*,
\end{equation}
with $\HH$ the Hilbert space of square integrable differential
forms and $D$ the Hodge-Dirac operator (self-adjoint on
a suitable domain).
The grading $\gamma\omega:=(-1)^k\omega$ on $k$-forms is
extended by linearity to $\HH$.
I will refer to \eqref{eq:HD} as the \emph{Hodge} spectral triple of $M$.
This spectral triple is even, even if $M$ is odd-dimensional, and is unital if{}f $M$ is compact.

It is well-known that, on pure states/points of $M$, the spectral distance coincides with the geodesic distance (with the convention that the geodesic distance between points in different connected components, if any, is infinite).

If $M$ is complete, the spectral distance between two arbitrary states coincides with the Wasserstein distance of order $1$ between the associated probability distributions (see e.g.~\cite{DM09}).
To explain the difficulty in computing such a distance (and then, more generally, the spectral distance), it is worth mentioning that the only case where the problem is completely solved is the real line \cite[\S3.1]{RR98}.

\subsection{A simple noncommutative example: the Bloch sphere}\label{sec:Bloch}
The Bloch sphere is a geometrical realization of the space of pure states of a two-level quantum mechanical system (or \emph{qubit}), cf.~\S2.5.2 and \S10.1.3 of \cite{RP11}. Among the several metrics used in quantum information, it is worth mentioning the Bures metric, which for a qubit can be explicitly computed and is given by Eq.~(9.48) of \cite{BZ06}. Here I will discuss some natural metrics coming from spectral triples.

So, let $\A=M_2(\C)$ be our algebra of ``observables''.
We cannot use the obvious representation on $\C^2$, since there is not enough room for a Dirac operator giving a finite distance \cite{IKM01}. Take \mbox{$\HH=M_2(\C)$} with Hilbert-Schmidt inner product $\inner{a,b}_{\mathrm{HS}}=\frac{1}{2}\,\tr(a^*b)$.\footnote{Note the different normalization between $\inner{\,,\,}_{\mathrm{HS}}$, here and in \S\ref{sec:qWass}, and $\inner{\,,\,}_{\mathrm{Tr}}$ in \eqref{eq:innerHS}.} The basis of Pauli matrices:
$$
\sigma_0=1 \;,\qquad
\sigma_1=\left[\begin{array}{cc} 0 & 1 \\ 1 & 0 \end{array}\right] \;,\qquad
\sigma_2=\left[\begin{array}{cr} 0 & \!-i \\ i & 0 \end{array}\right] \;,\qquad
\sigma_3=\left[\begin{array}{cr} 1 & 0 \\ 0 & \!-1 \end{array}\right] \;,
$$
is orthonormal for such a product.
The algebra $\A$ is represented on $\HH$ by left multiplication.

Every state on $\A$ is normal, with density matrices in bijection with points of the closed unit ball $B_3=\{x\in\R^3:\|x\|_2\leq 1\}$
in $\R^3$. Indeed, any selfadjoint matrix with trace $1$ has the form
\begin{equation}\label{eq:densityBall}
\rho_x=\tfrac{1}{2}(1+\textstyle{\sum_{i=1}^3}x_i\sigma_i) \;,
\end{equation}
where $x=(x_1,x_2,x_3)\in\R^3$. Since its two eigenvalues are
$\frac{1}{2}(1\pm\|x\|_2)$, it is clearly positive if{}f $\|x\|_2\leq 1$. It is a projection, hence a pure state, if $\|x\|_2=1$.
The map
$$
B_3\to\mc{S}(\A) \;,\qquad
x\mapsto \tr(\rho_x\;.\,) \;,
$$
is then the bijection we were looking for (in fact, a homeomorphism with respect to the weak$^*$ topology). We will identify $\mc{S}(\A)$ and $B_3$.

We can retrieve the Euclidean distance by considering the
Dirac operator
$$
D(a):=a^* \;,\qquad\forall\;a\in\HH\,.
$$
We can think of $(\A,\HH,D)$ as a spectral triple over the real numbers.
As a real Hilbert space, $\HH$ has inner product $\inner{a,b}_{\R}=\Re\inner{a,b}_{\mathrm{HS}}$, Pauli matrices are still orthonormal, the representation of $\A$ is still a $*$-representation , and $D$ is selfadjoint in the sense of $\R$-linear operators:
$\inner{a,D(b)}_{\R}=\inner{D(a),b}_{\R}\;\forall\;a,b\in\HH$.

\begin{prop}
For all $x,y\in B_3$, $d_D(x,y)=\|x-y\|_2$.
\end{prop}

\begin{proof}
Let $a=a_01+\sum_iv_i\sigma_i$,
with $a_0\in\R$ and $v=(v_1,v_2,v_3)\in\R^3$ and note that
$
\tr\big\{(\rho_x-\rho_y)a\big\}=\inner{x-y,v}
$
is the canonical inner product on $\R^3$, and $d_D(x,y)$ is
its sup under the condition $\|[D,a]\|\leq 1$.
Since $D$ is an isometry, $[D,a]$ and $D[D,a]$ have the same norm. But $D[D,a]b=[a,b]$ and
$$
\|[D,a]\|^2=\sup_{b\neq 0}\frac{\inner{[a,b],[a,b]}_{\R}}{\inner{b,b}_{\R}} \;.
$$
Writing $a=a_01+\sum_iv_i\sigma_i$ as above and $b=b_01+\sum_iw_i\sigma_i$, $b_0\in\C$, $w\in\C^3$:
$$
\|[D,a]\|^2=\sup_{b\neq 0}\frac{\|v\wedge w\|_2^2}{|b_0|^2+\|w\|^2_2} \;.
$$
The maximum is reached when $b_0=0$ and $w$ is orthogonal to $v$,
and we get $\|[D,a]\|=\|v\|_2$. Therefore
$$
d_D(x,y)=\sup_{v\in\R^3}\big\{ \inner{x-y,v}:\|v\|_2\leq 1 \big\}\;,
$$
which is equal to $\|x-y\|_2$ by Cauchy-Schwarz inequality.
\end{proof}

We can get the same metric from a complex-linear Dirac operator, if we use a degenerate representation of the algebra. Let $\A$ as above, with representation on $\widetilde{\HH}=\C^2\oplus\C^2$ given by $\widetilde{\pi}(a):=a\oplus 0$ (so, acting by matrix multiplication on the first summand, and trivially on the second), and let $\widetilde{D}(v\oplus w)=w\oplus v$ be the flip.

\begin{prop}
For all $x,y\in B_3$, $d_{\widetilde{D}}(x,y)=\|x-y\|_2$.
\end{prop}
\begin{proof}
Since $D[D,\pi(a)]=a\oplus -a\;\forall\;a\in\A$, $\|[D,a]\|$ is the matrix norm of $a$. Writing as before $a=a_01+\sum_iv_i\sigma_i$, with $a_0\in\R$ and $v=(v_1,v_2,v_3)\in\R^3$, we find
$\|a\|=|a_0|+\|v\|_2$. By Cauchy-Schwarz inequality, $\inner{x-y,v}\leq\|x-y\|_2\|v\|_2$, and the inequality is saturated if $a_0=0$ and $v=(x-y)/\|x-y\|_2$.
\end{proof}

Of the two spectral triples above, one was over the field of real numbers, and one was not unital ($\widetilde{\pi}(1)$ is not the identity on $\widetilde{\HH}$).
Another way to get the Euclidean distance, from a proper unital spectral triple (over $\C$), is via the the $SU(2)$-equivariant Dirac operator of the fuzzy sphere \cite[Prop.~4.3]{DLV12}.

Still another unital spectral triple on $M_2(\C)$ is the one obtained as a truncation of Moyal spectral triple \cite[\S4.3]{CDMW09}, which I briefly describe in the following.

Let $\A=M_2(\C)$ and $\widehat{\HH}=\C^2\oplus\C^2$ as before, but with representation $\widehat{\pi}(a):=a\oplus a$, and let $\widehat{D}$ be given by
$\widehat{D}(v\oplus w)=D_-(w)\oplus D_+(v)$, where
$$
D_+=\left[\begin{array}{cc} 0 & 1 \\ 0 & 0 \end{array}\right]
$$
and $D_-=(D_+)^*$. In this case the computation of $d_{\widehat{D}}$
is a bit more involved.

\begin{prop}[Prop.~4.4 of \cite{CDMW09}]
For all $x,y\in B_3$, with $x\neq y$, one has $d_{\widehat{D}}(x,y)=c(\theta)\|x-y\|_2$, where $\theta\in[0,\pi]$ is the polar angle of $x-y$ and
$$
c(\theta)=
\left\{\!\begin{array}{ll}
\;\sin\theta & \text{if }\;\frac{\pi}{4}\leq\theta\leq\frac{3\pi}{4},\\
|2\cos\theta|^{-1} & \text{otherwise}.
\end{array}\right.
$$
\end{prop}

\begin{proof}

As in the proof of Lemma \ref{lemma:DaC3}, $\|[\widehat{D},a]\|=\|[D_+,a]\|$ for all $a\in\mc{\A}^{\mathrm{sa}}$.
For $a=a_01+\sum_iv_i\sigma_i$, with $a_0\in\R$ and $v=(v_1,v_2,v_3)\in\R^3$,
$$
\|[D_+,a]\|=\|v\|_2+|v_3| \;.
$$
Due to a rotational symmetry on the horizontal plane,  we see that 
$d_{\widehat{D}}(x,y)$ is the sup of
$$
\inner{x-y,v}=r t (\cos\theta \cos\phi+\sin\theta \sin\phi)
$$
where $r=\|x-y\|_2$, $t=\|v\|_2$, $\phi$ is the polar angle of $v$ and $\theta$ the one of $x-y$. The norm constraint is
$t(1+|\cos\phi|)\leq 1$.
The factor $c(\theta)$ in the proposition is then given by
$$
c(\theta)=\sup_{0\leq\phi\leq \pi/2}\frac{|\cos\theta| \cos\phi+\sin\theta\sin\phi}{1+\cos\phi} \;.
$$
In the range $[0,\frac{\pi}{2}]$, the derivative vanishes only if
$\tan\frac{\phi}{2}=|\tan\theta|$, i.e.~i)~$\phi=2\theta$ if $\theta\in[0,\frac{\pi}{4}]$, ii)~$\phi=2(\pi-\theta)$ if $\theta\in[\frac{3\pi}{4},\pi]$, and iii)~it is always non-negative if $\frac{\pi}{4}<\theta<\frac{3\pi}{4}$. In case (i) and (ii), the stationary point is a maximum, and gives $c(\theta)=|2\cos\theta|^{-1}$. In case (iii), the sup is for $\phi=\pi/2$, and we get $c(\theta)=\sin\theta$.
\end{proof}

\subsection{Wasserstein distance between quantum states}\label{sec:qWass}
We can define a Wasserstein metric on quantum states by means of Berezin quantization.

Let $G$ be a compact group and $U:G\to A:=\B(\HH)$ a finite-dimensional unitary irreducible representation of $G$. The adjoint action $\alpha$ on $A$ is defined by $\alpha_x(a)=U_xaU_x^*$. I will denote by $\inner{a,b}_{\mathrm{HS}}=\frac{1}{N}\,\tr(a^*b)$ the Hilbert-Schmidt inner product on $A$, with $N:=\dim_{\C}(\HH)$, and by $\|a\|_{\mathrm{HS}}=\inner{a,a}_{\mathrm{HS}}^{1/2}$ the corresponding norm. Let $P\in\B(\HH)$ be a density matrix, $H\subset G$ the stabilizer of $P$:
$$
H:=\{x\in G:\alpha_x(P)=P\} \;,
$$
and $M=G/H$ the quotient space. I will denote by $\de\mu$ the $G$-invariant measure on $M$, normalized to $1$, and by
$$
\inner{f,g}_{L^2}:=\int_Mf(x)^*g(x)\de\mu_x
$$
the inner product on $L^2(M,\de\mu)$, where functions on $M$ are identified with right $H$-invariant functions on $G$.
We can define two linear maps, a \emph{symbol map} $A\to C(M)$,
$a\mapsto\sigma_a$, and a \emph{quantization map} $C(M)\to A$,
$f\mapsto\mc{Q}_f$, as follows:
$$
\sigma_a(x):=\tr\big(\alpha_x(P)a\big) \;,\qquad\quad
\mc{Q}_f:=N\int_M\alpha_x(P)f(x)\de\mu_x
$$
for all $a\in A$ and $f\in C(M)$. They are one the adjoint of the other, that is
$$
\inner{\sigma_a,f}_{L^2}=\inner{a,\mc{Q}_f}_{\mathrm{HS}}
 \;,\qquad\quad\forall\;a\in A,f\in C(M),
$$
as one can easily check. Since $\alpha_x(P)$ is a density matrix, the map $\sigma$ is unital, positive and norm non-increasing. With a little work one can prove that the same three properties hold for $\mc{Q}$ (see e.g.~\cite{Rie04}).
In particular, the operator
\begin{equation}\label{eq:bothsides}
N\int_M\alpha_x(P)\de\mu_x
\end{equation}
is $\alpha$-invariant, i.e.~commutes with $U_y$ for all $y\in G$. Since $U$ is irreducible, from Schur's lemma \eqref{eq:bothsides} is proportional to $1_A$. The proportionality constant can be computed by taking the trace of \eqref{eq:bothsides}, thus proving that $\mc{Q}$ maps $1$ to $1_A$.

Due to the above properties, $\frac{1}{N}\mc{Q}$ sends probability measures into density matrices, and $N\sigma$ sends density matrices into probability distributions\footnote{In particular,
$\int_MN\sigma_\rho(x)\de\mu_x=N\inner{\sigma_\rho,1}_{L^2}=
N\inner{\rho,1}_{\mathrm{HS}}
=\tr(\rho)=1$ for any $\rho$.}. The latter map is injective under the following assumptions: suppose $G$ is a connected compact semisimple Lie group and $P$ a rank-one projection (a pure state) which has the highest weight vector of the representation in its range. Then $\mathrm{Span}\{\alpha_x(P):x\in G\big\}=\B(\HH)$
\cite[Thm.~3.1]{Rie04}, the symbol map is injective\footnote{$\sigma_a(x)=\inner{N\alpha_x(P),a}_{\mathrm{HS}}=0\;\forall\;x$ implies $\inner{b,a}_{\mathrm{HS}}=0\;\forall\;b\in\B(\HH)$, and then $a=0$.} and the quantization map -- being its adjoint -- is surjective.

Now, with a surjective symbol map we can give the following definition. Note that on $M$ there is a unique $G$-invariant Riemannian metric, with normalization fixed by the condition that the associated volume form is $\de\mu$.

\begin{df}
We call \emph{cost-distance} $W(\rho,\tau)$ between two density matrices the Wasserstein distance of order $1$ between the two probability distributions $N\sigma_\rho$ and $N\sigma_\tau$, with cost given by the geodesic distance on $M$.
\end{df}

Note that the map $\rho\mapsto N\sigma_\rho$ coincides with the map $\mc{S}(A)\to\mc{S}(C(M))$ obtained by pulling back states with the quantization map that we considered in \S6 of \cite{DLM13}. This follows from the identity $\inner{N\sigma_\rho,f}_{L^2}=\tr(\rho\hspace{1pt}\mc{Q}_f)$.

\smallskip

That $W$ is finite follows from the next proposition.

\begin{prop}
$W(\rho,\tau)\leq N^{3/2}\ell \|\rho-\tau\|_{\mathrm{HS}}$
for all density matrices $\rho$ and $\tau$, where
$$
\ell:=\int_Md_{\mathrm{geo}}(e,x)\de\mu_x
$$
is the mean value of the geodesic distance.
\end{prop}

\begin{proof}
$W(\rho,\tau)$ is the sup over all $1$-Lipschitz functions $f$ of:
\begin{align*}
N\int_M\big\{\sigma_\rho(x)-\sigma_\tau(x)\big\}f(x)\de\mu_x &=
N^2\int_M\inner{\alpha_x(P),\rho-\tau}_{\mathrm{HS}}f(x)\de\mu_x
\\
&=N^2\int_M\inner{\alpha_x(P),\rho-\tau}_{\mathrm{HS}}\big\{f(x)-f(e)\big\}\de\mu_x \;.
\end{align*}
Here $e\in M$ is the class of the unit element of $G$, and we used the fact that $N\int_M\alpha_x(P)\de\mu_x=1$, and $\inner{1,\rho-\tau}_{\mathrm{HS}}=N^{-1}(\tr\,\rho-\tr\,\tau)=0$. 
Note also that $\|\alpha_x(P)\|_{\mathrm{HS}}=N^{-\frac{1}{2}}$, since $\alpha_x(P)$ is a rank $1$ projection. Hence from Cauchy-Schwarz inequality and the $1$-Lipschitz condition:
$$
W(\rho,\tau)\leq N^{3/2}	\|\rho-\tau\|_{\mathrm{HS}}\int_Md_{\mathrm{geo}}(e,x)\de\mu_x \;. \vspace{-18pt}
$$
\end{proof}

Let $\rho_x:=\alpha_x(P)$ be the density matrix associated to a coherent state, depending on the class $x\in M$. The map $x\mapsto\rho_x$ is a homeomorphism from $M$ to the set of coherent states with weak$^*$ topology.

\begin{prop}
If $G$ is abelian, then $W(\rho_x,\rho_y)\leq d_{\mathrm{geo}}(x,y)\;\forall\;x,y\in M$.
\end{prop}

\begin{proof}
$W(\rho_x,\rho_y)$ is the sup over all $1$-Lipschitz functions $f$ of:
\begin{align*}
W_f(x,y) &:=
N^2\int_M\inner{\alpha_z(P),\rho_x-\rho_y}_{\mathrm{HS}}f(z)\de\mu_z
\\[2pt]
&=N^2\int_M\inner{(\alpha_{x^{-1}z}-\alpha_{y^{-1}z})P,P}_{\mathrm{HS}}f(z)\de\mu_z
\\[2pt]
&=N^2\int_M\inner{\alpha_z(P),P}_{\mathrm{HS}}\big\{f(xz)-f(yz)\big\}\de\mu_z \;.
\end{align*}
Writing $P=\xi\inner{\xi,\,.\,}_{\HH}$, with $\xi\in\HH$ a unit vector in the range of $P$, we deduce that
$$
\inner{\alpha_z(P),P}_{\mathrm{HS}}=\left|\inner{U_z\xi,\xi}\right|^2
$$
is non-negative. Also
$$
N^2\int_M\inner{\alpha_z(P),P}_{\mathrm{HS}}\de\mu_z=
N\inner{1_A,P}_{\mathrm{HS}}=\tr(P)=1 \;.
$$
From the $1$-Lipschitz condition we get:
$$
W_f(x,y) \leq N^2\int_M\inner{\alpha_z(P),P}_{\mathrm{HS}}d_{\mathrm{geo}}(xz,yz)\de\mu_z
  \;.
$$
Here is where we need the hypothesis that $G$ is abelian: the geodesic distance is \emph{left} $G$-invariant, so if $G$ is abelian $d_{\mathrm{geo}}(xz,yz)=d_{\mathrm{geo}}(x,y)$ and we get:
$$
W_f(x,y) \leq d_{\mathrm{geo}}(x,y)\cdot N^2\int_M\inner{\alpha_z(P),P}_{\mathrm{HS}}\de\mu_z=d_{\mathrm{geo}}(x,y) \;. \vspace{-15pt}
$$
\end{proof}

\begin{ex}
If $G=SU(2)$ and $\HH=\C^2$ is the defining representation, one can explicitly compute the cost-distance $W$ between arbitrary density matrices: with the identification $\mc{S}(M_2(\C))\simeq B_3$, one finds that $W$ is proportional to the Euclidean distance on the unit ball.

If $G=SU(N)$ and $\HH=\C^N$ is the defining representation, the set of coherent states and the set of pure states coincide (both are isomorphic to $M=\C\mathrm{P}^{N-1}$).
\end{ex}

Since $\mc{Q}:C(M)\to A$ is surjective, we can define a quotient seminorm $L$ on $A$ as:
$$
L(a):=\inf_{f\in C(M)}\big\{\|f\|_{\mathrm{Lip}} :\mc{Q}_f=a \big\}\,,
$$
where $\|f\|_{\mathrm{Lip}}=\sup_{x\neq y}|f(x)-f(y)|/d_{\mathrm{geo}}(x,y)$ is the Lipschitz seminorm. Since $\inner{\sigma_\rho,f}_{L^2}=\inner{\rho,\mc{Q}_f}_{\mathrm{HS}}$ for any density matrix $\rho$, we get
$$
W(\rho,\tau)=\sup_{a\in A^{\mathrm{sa}}}\big\{
\inner{\rho-\tau,a}_{\mathrm{HS}}:L(a)\leq 1\big\} \;.
$$
It would be nice to prove that the seminorm $L$ comes from a spectral triple. A possible route to spectral triples is by extending the quantization map from functions on $M$ to spinors. The case $G=SU(2)$ is discussed in \cite{DLV12} and in \S6.4 of \cite{DLM13}. We can interpret the first inequality in Prop.~6.16 of \cite{DLM13} as follows: the spectral distance associated to the natural Dirac operator bounds the cost-distance from above.

\section{Pythagoras equality for commutative spectral triples}\label{sec:commex}

In this section, I collect some (commutative) examples of products for which Pythagoras equality holds.
Further examples, including a noncommutative one (Moyal plane), are briefly discussed in the next section (cf.~Cor.~\ref{cor:7.3}).

\subsection{Pythagoras for a product of Riemannian manifolds}\label{sec:Hodge}

Given the Hodge spectral triples of two manifolds $M_1$ and $M_2$, we can define two spectral triples on $M=M_1\times M_2$. One as product of the spectral triples of the two factors, and one as the Hodge spectral triple associated to the product Riemannian metric on $M$ cf.~\eqref{eq:lineel}. These two give the same distance, cf.~\S3.2 of \cite{DM13}.

Verifying Pythagoras equality for the product of the Hodge spectral triples of two Riemannian manifold is then reduced to the problem of proving that the product metric on $M$ in the sense of Riemannian manifold, i.e.~\eqref{eq:lineel}, induces  the product distance in the sense of metric spaces, cf.~\S\ref{sec:2.1}.
This can be proved as follows: given two points $x=(x_1,x_2),y=(y_1,y_2)\in M_1$ 
let $t\mapsto z(t)=(z_1(t),z_2(t))$ be a geodesic between $x$ and $y$,
parametrized by its proper length (it is enough to give the proof when $x,y$ are in the same connected component). One can prove that
$t\mapsto z_1(t)$ is a geodesic in $M_1$ between $x_1$ and $y_1$, and similarly $t\mapsto z_2(t)$ is a geodesic in $M_2$ between $x_2$ and $y_2$, cf. \S3.1 of \cite{DM13}, and that $t$ is an affine parameter
(not necessarily the proper length) for both curves. Integrating the line element one then proves that:
$$
d_D(x,y)^2=d_{D_1}(x_1,y_1)^2+d_{D_2}(x_2,y_2)^2 \;,
$$
where $d_{D_i}$ is the spectral/geodesic distance on $M_i$, and $d_D$ is the one of the Cartesian product. So, the product of Hodge spectral triples is orthogonal in the sense of Def.~\ref{def:2.2}.

In the proof, it is crucial the use of geodesics. In \S\ref{sec:CN}, I will show how to give a completely algebraic proof of Pythagoras equality for arbitrary finite metric spaces.

Let me also stress that, already in the example of Riemannian manifolds, Pythagoras equality doesn't hold for arbitrary states. In \S 3.3 of \cite{DM13} we discuss a simple example where the ratio $d_D/d_{D_1}\!\boxtimes d_{D_2}$ assumes all possible values between $1$ and $\sqrt{2}$.

\subsection{Product of two-point spaces}\label{sec:C2}
The simplest possible example is the product of two copies of the two-point space spectral triple discussed in \S\ref{sec:twopoints}.
Let then $(\A_1,\HH_1,D_1,\gamma_1)=(\C^2,\C^2,F,\gamma)$ and $(\A_2,\HH_2,D_2)=(\C^2,\C^2,F)$.
To simplify the discussion, we fixed to $1/2$ the parameter
$\lambda$ in \S\ref{sec:twopoints}.
We now prove that:

\begin{prop}\label{lemma:C2}
For all product states $\varphi,\psi$ and all
$a\in\A^{\mathrm{sa}}$, the condition \eqref{eq:ineqD} is satisfied.
So: Pythagoras equality holds for arbitrary product states.
\end{prop}

\begin{proof}
Product states $\varphi=\varphi_1\otimes\varphi_2$ are in bijection with pairs $(\varphi_1,\varphi_2)\in [-1,1]\times [-1,1]$, cf.~\eqref{eq:stateC2}.
Given two product states $\varphi,\psi$, $\ker(P_{\varphi,\psi})$ is spanned by the vector:
$$
v_{\varphi,\psi}:=
\left[\!\begin{array}{cc} \varphi_1+1 & 0 \\ 0 & \varphi_1-1 \end{array}\!\right]
\otimes
\left[\!\begin{array}{cc} \psi_2+1 & 0 \\ 0 & \psi_2-1 \end{array}\!\right] \in\A^{\mathrm{sa}} \;.
$$
We can decompose any $a\in\A^{\mathrm{sa}}$ as
$$
a=x_0\hspace{1pt}1\otimes 1
+\frac{x_1}{2}\hspace{1pt}\gamma\otimes 1
+\frac{x_2}{2}\hspace{1pt}1\otimes\gamma
+\frac{x_3}{2}\hspace{1pt}v_{\varphi,\psi} \;,\qquad x_0,\ldots,x_3\in\R,
$$
where the first three terms span $\A_1+\A_2$: they are linearly independent and $\A_1+\A_2$ has dimension $\dim\A-\dim\ker(P_{\varphi,\psi})=3$.
A simple computation gives
\begin{equation}\label{eq:mat}
-(\gamma\otimes\gamma)[D,a]=F\otimes\big\{(x_1+x_3\psi_2)\gamma+x_3\big\}
+\big\{(x_2+x_3\varphi_1)+x_3\gamma\big\} \otimes F
\;,
\end{equation}
were in the computation we noticed that
$$
v_{\varphi,\psi}=
(\varphi_1\,1+\gamma)\otimes
(\psi_2\,1+\gamma) \;.
$$
With the isomorphism of unital $*$-algebras $M_2(\C)\otimes M_2(\C)\to M_4(\C)$
(the right factor inserted in $2\times 2$ blocks) we transform \eqref{eq:mat}
into the matrix:
\begingroup
\addtolength{\arraycolsep}{-5pt}
$$
L:=\left[\;\begin{matrix}
0 & (x_2+x_3\varphi_1)+x_3 & x_3+(x_1+x_3\psi_2) & 0 \\
(x_2+x_3\varphi_1)+x_3 & 0 & 0 & x_3-(x_1+x_3\psi_2) \\
x_3+(x_1+x_3\psi_2) & 0 & 0 & (x_2+x_3\varphi_1)-x_3 \\
0 & x_3-(x_1+x_3\psi_2) & (x_2+x_3\varphi_1)-x_3 & 0
\end{matrix}\;\right] .
$$
\endgroup
The eigenvalues of $L^2$ can be computed by first writing the characteristic polynomial and then solving a degree $2$ equation (or call $a=x_1+x_3\psi_2$, $b=x_2+x_3\varphi_1$,
$c=x_3$ and compute the norm as a function of $a,b,c$ with Mathematica$^\copyright$).
In this way we get $\|L\|=\|[D,a]\|$, that, as a function of $x_1,x_2,x_3$, is given by:
\begin{equation}\label{eq:fdix}
f(x):=\sqrt{2}\,|x_3|+\sqrt{(x_1+x_3\psi_2)^2+(x_2+x_3\varphi_1)^2}
\end{equation}
The norm $\|[D,P_{\varphi,\psi}(a)]\|$ is obtained from \eqref{eq:fdix} the substitution $x_3\to 0$.
From the triangle inequality $|a+b|\geq |a|-|b|$,
$$
f(x)\geq\sqrt{2}\,|x_3|+\sqrt{(|x_1|-|x_3|)^2+(|x_2|-|x_3|)^2} \;.
$$
Called $t:=|x_3|$, one can check that $\partial f/\partial t$ is non-negative, so
$f$ is an increasing function of $|x_3|$. Thus
$f(x)\geq f|_{x_3=0}$, which is what we wanted to prove.
\end{proof}

Every state of $\C^2$ is normal.
We can use this simple example to show that the upper bound in Prop.~\ref{prop:9} is not optimal.
Next proposition is independent of the choice of density matrices, i.e.~on how we extend states of $\C^2$ to $M_2(\C)$.

\begin{prop}\label{prop:11}
For $\varphi_1=\psi_2=1$, $K=\|P_{\varphi,\psi}\|=3$.
\end{prop}

\begin{proof}
Due to Lemma \ref{lemma:K3}, it is enough to prove that $K\geq 3$.
For $a=\gamma\otimes\gamma$:
$$
P_{\varphi,\psi}(a)=1\otimes 1-\gamma\otimes 1-1\otimes\gamma \;.
$$
Since $\|a\|=1$,
one has $
K\geq \|P_{\varphi,\psi}(a)\|.
$
For $v=(0,1)^t\otimes (0,1)^t$, $P_{\varphi,\psi}(a)v=3v$.
Hence $\|P_{\varphi,\psi}(a)\|\geq 3$.
\end{proof}

\begin{rem}
Prop.~\ref{lemma:C2} remains valid if we rescale the Dirac operators by arbitrary scale factors, say $D_1=\lambda F$ and $D_2=\mu F$, with $\lambda,\mu>0$.
In this case \eqref{eq:fdix} is replaced by
\begin{align*}
f(x) &:=\sqrt{\lambda^2+\mu^2}\,|x_3|+\sqrt{\lambda^2(x_1+x_3\psi_2)^2+\mu^2(x_2+x_3\varphi_1)^2} \\
&\geq
\sqrt{\lambda^2+\mu^2}\,|x_3|+\sqrt{\lambda^2(|x_1|-|x_3|)^2+\mu^2(|x_2|-|x_3|)^2}
 \;,
\end{align*}
(by the triangle inequality) and with a derivation with respect to $t=|x_3|$ we prove that $f$ is a non-decreasing function of $t$, so
$f(x)\geq f|_{x_3=0}=\|[D,P_{\varphi,\psi}(a)]\|$.
\end{rem}

\subsection{Product of finite metric spaces}\label{sec:CN}
Here we consider two arbitrary finite metric spaces $X_1$
and $X_2$, with $N_1$ resp.~$N_2$ points, and the product of the corresponding canonical spectral triples introduced in \S\ref{sec:finite}. We adopt the notations of \S\ref{sec:finite}, and distinguish the two spectral triples by a sub/super-script $1,2$.

Their product is a direct sum,
over all $i\neq j$ and $k\neq l$, of the spectral triples:
\begin{equation}\label{eq:summandA}
\Big(\, \A_1\otimes\A_2
\,,\,
\HH_{ij}^1\otimes\HH_{jk}^2\,,\,
\pi_{ijkl}:=\pi_{ij}^1\otimes \pi_{kl}^2 \,,\,
D_{ijkl}:=D^1_{ij}\otimes 1+\gamma_{ij}^1\otimes D^2_{kl} \,\Big)
\end{equation}
which in turn is the product of the triples
\begin{equation}\label{eq:summandB}
( \A_1,\HH_{ij}^1=\C^2,\pi_{ij}^1,D^1_{ij},\gamma^1_{ij} ) \;,\qquad\quad
( \A_2,\HH_{kl}^2=\C^2,\pi_{kl}^2,D^2_{kl} ) \;,
\end{equation}
In this example, Pythagoras equality is satisfied by arbitrary pure states.

\begin{thm}
Given two arbitrary finite metric spaces, the product of their canonical spectral triples is orthogonal in the sense of Def.~\ref{def:2.2}.
\end{thm}

\begin{proof}
With the pure states \eqref{eq:xxstar} one can construct morphisms:
$$
x^s_i\otimes x^s_j:\C^{N_s}\to\C^2 \;,\qquad
a\mapsto (a_i,a_j) \;,
$$
where $s=1,2$.
Each triple \eqref{eq:summandB} is the pullback of the canonical spectral triple on $\C^2$ in \S\ref{sec:twopoints} (possibly with different normalizations of the Dirac operator).
Assume that $\varphi$ and $\psi$ are pure:
$$
\varphi=x_r^1\otimes x_p^2 \;,\qquad\quad
\psi=x_s^1\otimes x_q^2 \;,
$$
with $r,p,s,q$ fixed.
Since \eqref{eq:summandA} is a direct sum,
$$
\|[D,a]\|=\max_{i\neq j,k\neq l}\|[D_{ijkl},\pi_{ijkl}(a)]\|\geq \|[D_{rspq},\pi_{rspq}(a)]\|
$$
and
\begin{equation}\label{eq:finiteineq}
d_D(\varphi,\psi)\leq d_{D_{rspq}}(\varphi,\psi) \;.
\end{equation}
In this way, we reduce the problem to a product of two-point spaces.
As shown in \S\ref{sec:C2}, for a product of two-point spaces
$$
d_{D_{rspq}}(\varphi,\psi)=\sqrt{ (g^1_{rs})^2+(g^2_{pq})^2 }=
\sqrt{ d_{D_1}(\varphi_1,\psi_1)^2+d_{D_2}(\varphi_2,\psi_2)^2 } \;.
$$
This proves $d_D(\varphi,\psi)\leq
d_{D_1}\!\boxtimes d_{D_2}(\varphi,\psi)$, the opposite inequality being
always true, cf.~\eqref{eq:ineqDtimes}.
\end{proof}

\section{Pythagoras from generalized geodesics}\label{sec:geo}

In this section, I will discuss a class of states generalizing pure states of a complete Riemannian manifold
and translated states on Moyal plane. For such states, I will then prove Pythagoras equality for the product of an arbitrary spectral triple with the two-point space, cf.~Prop.~\ref{prop:prodC2}, generalizing the results in \cite{MW02,MT13}. I will then show how to extend the result from the two-point space to an arbitrary finite metric space, cf.~Prop.~\ref{prop:C2toCN}.

\subsection{Geodesic pairs}

Let $\sigma:\R\to\mathrm{Aut}(A)$ be a strongly continuous one-parameter group of $*$-automorphisms of a $C^*$-algebra $A$.
Let $\mc{D}(\delta)$ be the set of all $a\in A$ for which the norm limit:
$$
\delta(a):=\lim_{t\to 0}\frac{\sigma_t(a)-a}{t}
$$
exists. Then $\mc{D}(\delta)$ is a dense $*$-subalgebra of $A$ and $\delta$ is a (unbounded) closed \mbox{$*$-derivation} with domain $\mc{D}(\delta)$ \cite[\S3]{Sak91}. Conversely, one can give sufficient conditions for $\delta$ to generate an action of $\R$ on $A$ \cite[\S3.4]{Sak91}.

Let $\A\subset\mc{D}(\delta)\subset A$ be a $*$-subalgebra (not necessarily dense).
To any state $\tau_0$ on $A$ we can associate a family of states:
\begin{equation}\label{eq:Tstates}
\tau_t(a):=\tau_0(\sigma_t(a)) \;.
\end{equation}
What allows to prove Pythagoras theorem in the examples in \cite{MW02,MT13} is that the curve $t\mapsto\tau_t$ in state space is parametrized by the arc length (or has ``unit speed'').

\begin{df}\label{def:gp}
Let $(\A,\HH,D,\gamma)$ be a spectral triple. A curve $t\mapsto\varphi_t$ in state space is called a \emph{metric straight line} if $d_D(\varphi_t,\varphi_s)=|t-s|$ for all $t,s\in\R$ \cite[\S6.1]{DD09}.
\linebreak
Let $I\subset\R$ be an interval containing $0$.
In the notations above,
$(\sigma|_I,\tau_0)$ will be called a \emph{geodesic pair} if
$d_D(\tau_t,\tau_s)=|t-s|\;\forall\;t,s\in I$.
\end{df}

\begin{ex}\label{ex:M}
Consider the Hodge-Dirac spectral triple of a complete oriented Riemannian manifold $(M,g)$, or the natural spectral triple of a complete Riemannian spin manifold. In both cases, one easily checks that (see e.g.~\cite{DM09} for the latter example):
\begin{equation}\label{eq:sg1}
[D,f]^*[D,f]=g^{\mu\nu}\partial_\mu f\partial_\nu f
\;,\qquad\forall\;f\in C^\infty_0(M,\R) \;,
\end{equation}
where we use Einstein's convention of summing over repeated indexes.

By the Hopf-Rinow theorem, any two points are connected by a geodesic of minimal length, and every geodesic can be extended indefinitely (but it is only locally a metric straight line, hence the need of the interval $I$ in the definition above). Let
$$
c:\R\to M
$$
be any geodesic and $I\ni\{0\}$ a closed interval where $c$ is of minimal length. If the curve is parametrized by the proper length, $|\dot c(t)|=1\;\forall\;t$. Since $c(I)\subset M$ is a properly embedded submanifold, the vector field $\dot c$ along $c$ can be extended to a globally defined vector field $V$ on $M$ (cf.~Prop.~5.5 and exercise 8-15 of \cite{Lee13}). Clearly $|V_x|=|\dot c(t)|=1$ if $x=c(t)\in c(I)$. We can choose $V$ such that $|V_x|\leq 1$ for all $x\in M$ (use a parallel frame to define such a $V$ in a neighborhood of $c(I)$, and a bump function to extend it globally as the zero vector field outside $c(I)$).

Thinking of $V$ as a derivation on $f\in C^\infty_0(M)$, we define:
$$
\sigma_t(f)(x):=e^{tV}f(x) \;.
$$
In the notations above, $\delta(f)=Vf$
and by Cauchy-Schwarz inequality:
\begin{equation}\label{eq:sg2}
|\delta(f)(x)|^2\leq \|\partial_\mu f(x)\partial^\mu\|_2\cdot\|V_x\|_2\leq\|\partial_\mu f(x)\partial^\mu\|_2
=(g^{\mu\nu}\partial_\mu f\partial_\nu f)(x) \;,
\end{equation}
for all $f\in C^\infty_0(M,\R)$ and $x\in M$.
We will need \eqref{eq:sg1} and \eqref{eq:sg2} later on.

If $\tau_0$ is the pure state $\tau_0(f)=f(c(0))$, clearly $\tau_t(f)=f(c(t))\;\forall\;t\in I$, and
$$
d_D(\tau_t,\tau_s)=d_{\mathrm{geo}}(c(t),c(s))=|t-s| \;,
$$
for all $t,s\in I$. Hence $(\sigma|_I,\tau_0)$ is a geodesic pair.
\end{ex}

\begin{lemma}\label{lemma:msl}
$d_D(\tau_t,\tau_s)\leq |t-s|\;\forall\;t,s
\iff \tau_t\big(\delta(a)\big)\leq\|[D,a]\|\;\forall\;t$ and $a\in\A^{\mathrm{sa}}$.
\end{lemma}

\begin{proof}
$d_D(\tau_t,\tau_s)\leq |t-s|$ if{}f
\begin{equation}\label{eq:condL}
\tau_t(a)-\tau_s(a)\leq\|[D,a]\|\cdot|t-s|\quad\forall\;a\in\A^{\mathrm{sa}} .
\end{equation}
By symmetry, we can assume that $t>s$.
By linearity and continuity of $\tau_t$:
\begin{equation}\label{eq:inttau}
\tau_t(a)-\tau_s(a)=\int_s^t\frac{\de}{\de\xi}\,\tau_\xi(a)\de\xi=
\int_s^t\tau_\xi\big(\delta(a)\big)\de\xi \;.
\end{equation}
If $\tau_\xi\big(\delta(a)\big)\leq\|[D,a]\|$ for all $\xi$ in the interior $\mathring{I}$ of $I$, then
$$
\tau_t(a)-\tau_s(a)\leq\int_s^t \|[D,a]\| \de\xi =\|[D,a]\|\cdot|t-s|
$$
for all $t,s\in I$, proving ``$\Leftarrow$''. If, on the other hand \eqref{eq:condL} holds for all $t,s\in I$ and $a\in\A^{\mathrm{sa}}$, then by linearity and continuity of $\tau_t$:
$$
\tau_t\big(\delta(a)\big)=\lim_{t\to 0^+}\frac{\tau_{t+\epsilon}(a)-\tau_\epsilon(a)}{\epsilon}\leq\|[D,a]\|
$$
for all $t\in\mathring{I}$.
By continuity we can replace $\mathring{I}$ by $I$, proving ``$\Rightarrow$''.
\end{proof}

\begin{rem}
Note that everything works even if $\A$ is not dense in $A$. But in this case the map $\mc{S}(A)\to\mc{S}(\bar\A)$ is surjective but not bijective, and $\tau_0$ must be a state of $A$. In particular, given any unital spectral triple $(\A,\HH,D)$ and taking $A=\B(\HH)$, we can consider the action $\sigma_t(a)=e^{itD}ae^{-itD}$; the associated derivation is $\delta(a)=i[D,a]$ and automatically satisfies the condition of Lemma \ref{lemma:msl}.
\end{rem}

Unfortunately, in order to prove Pythagoras equality the condition in Lemma \ref{lemma:msl} is not enough. We need the following stronger assumption.

\begin{df}
A geodesic pair $(\sigma|_I,\tau_0)$ will be called \emph{strongly geodesic} if
\begin{equation}\label{eq:sg}
\tau_t\big((\delta a)^*(\delta a)+b^*b\big)
\leq\big\|[D,a]^*[D,a]+b^*b\big\|\;,\quad
\forall\;a,b\in\A^{\mathrm{sa}},t\in I.
\end{equation}
\end{df}

\begin{ex}
The geodesic pair in Example \ref{ex:M} is strongly geodesic. Indeed, for all real functions $a,b$ on $M$:
$$
\tau_t\big((\delta a)^*(\delta a)+b^*b\big)\leq
\big\|
g^{\mu\nu}\partial_\mu a\partial_\nu a+b^2
\big\|_\infty
=\|[D,a]^*[D,a]+b^*b\|
$$
where in the first inequality I used \eqref{eq:sg2}, and in the last equality I used \eqref{eq:sg1} and the observation that for functions the operator norm and the sup norm coincide.
\end{ex}

A noncommutative example is given by Moyal plane, which is recalled below.

\begin{lemma}\label{lemma:sgp}
Let $(\A,\HH\otimes\C^2,D)$ be an even spectral triple, with
obvious grading, trivial action of $\A$ on the $\C^2$ factor and
\begin{equation}\label{eq:Dpm}
D=\left[\begin{array}{cc}
0 & D_- \\ D_+ & 0
\end{array}\right] \;.
\end{equation}
Let $(\sigma|_I,\tau_0)$ be a geodesic pair. Assume
$$
\delta(a)=\tfrac{1}{2}[uD_+-\bar uD_-,a] \;,\qquad\forall\;a\in\A^{\mathrm{sa}}
$$
for some fixed $u\in U(1)$.
Then $(\sigma|_I,\tau_0)$ is strongly geodesic.
\end{lemma}

\begin{proof}
Since states have norm $1$,
it is enough to prove that
$$
\big\|(\delta a)^*(\delta a)+b^*b\big\|
\leq\big\|[D,a]^*[D,a]+b^*b\big\|
$$
for all $a,b\in\A^{\mathrm{sa}}$. The left hand side of this inequality is the operator norm of
$$
\left[\begin{array}{cc}
\delta a & 0 \\ b & 0
\end{array}\right]=
\frac{1}{2}\left[\begin{array}{cc}
u[D_+,a] & 0 \\ b & 0
\end{array}\right]+
\frac{1}{2}\left[\begin{array}{cc}
-\bar u[D_-,a] & 0 \\ b & 0
\end{array}\right] \;.
$$
From the triangle inequality:
$$
\mathrm{l.h.s.}\leq
\frac{1}{2}\left\|\left[\begin{array}{cc}
[D_+,a] & 0 \\ b & 0
\end{array}\right]\right\|+
\frac{1}{2}\left\|\left[\begin{array}{cc}
[D_-,a] & 0 \\ b & 0
\end{array}\right]\right\|=\frac{1}{2}\sum_{i=\pm}
\big\|[D_i,a]^*[D_i,a]+b^*b\big\|
 \;.
$$
On the other hand, from
$$
[D,a]^*[D,a]+b^*b=\left[\begin{array}{cc}
[D_+,a]^*[D_+,a]+b^*b & 0 \\ 0 & [D_-,a]^*[D_-,a]+b^*b
\end{array}\right]
$$
we get
$$
\big\|[D,a]^*[D,a]+b^*b\big\|=\max_{i\in\{+,-\}}
\big\|[D_i,a]^*[D_i,a]+b^*b\big\| \;,
$$
hence the thesis.
\end{proof}

\begin{ex}\label{ex:Moyal}
Let $(\A,\HH,D)$ be the even irreducible spectral triple of Moyal plane in \cite[\S3.3]{DLM13b}, given by $\A=\mc{S}(\N^2)$ the algebra of rapid decay matrices on $\HH=\ell^2(\N)\otimes\C^2$ with standard grading, and Dirac operator as in \eqref{eq:Dpm}. Here $D_+$ is given on the canonical orthonormal basis of $\ell^2(\N)$ by:
$$
D_+\ket{n}=\sqrt{\tfrac{2}{\theta}(n+1)}\ket{n+1} \;,
$$
for all $n\geq 0$, $D_-=(D_+)^*$ and $\theta>0$ is a deformation parameter.
Every state $\tau_0$ of $A=\mc{K}(\ell^2(\N))$ is normal,
$\tau_0(a)=\inner{\psi_0|a|\psi_0}$. For $z\in\C$,
we call
$$
T(z):=\exp\big\{\tfrac{1}{2}(zD_+-\bar zD_-)\big\} \;.
$$
The vectors $\ket{\psi_z}:=T(z)\ket{\psi_0}$
give a family of normal states $\Psi_z$, with $z\in\C$.
Fix $u,v\in \C$ with $|u|=1$, and let $t\mapsto z_t=u t+v$ be the corresponding line in $\C$. Then
$$
t\mapsto e^{\frac{i}{2\theta}\Im(u\bar v)}T(z_t)=
e^{-itX}T(v)
$$
is a  strongly continuous one-parameter group of unitaries generated by the (unbounded, selfadjoint on a suitable domain) operator
\begin{equation}\label{eq:XMoyal}
X:=\frac{i}{2}(uD_+-\bar uD_-) \;,
\end{equation}
and $\sigma_t(a)=T(z_t)^*aT(z_t)$
is a strongly continuous one-parameter group
of\linebreak $*$-automorphisms generated by the derivation
$\delta(a)=i[X,a]$.
In the notations above,
if $\tau_0:=\Psi_{z_0}$,
then $\tau_t=\Psi_{z_t}$ for all $t\in\R$;
from \cite[Prop.~4.3]{DLM13b} we get
$d_D(\tau_t,\tau_s)=|t-s|$ for all $t,s\in\R$ (the original proof is in \cite{MT13}). So, $(\sigma|_I,\tau_0)$ is a geodesic pair;
it follows from Lemma \ref{lemma:sgp} that it is also strongly geodesic.
\end{ex}

\subsection{Products by \texorpdfstring{$\C^2$}{C2}}
Let $(\A_1,\HH_1,D_1,\gamma_1)$ be the spectral triple of the two-point space in \S\ref{sec:twopoints},
$(\A_2,\HH_2,D_2)$ any unital spectral triples and
$(\A,\HH,D)$ the product spectral triple.
Let $(\sigma|_I,\tau_0)$ be a strongly geodesic pair on the second spectral triple, and $\tau_t^2$ the corresponding state. Finally, let
$$
[-\lambda,\lambda]\to\mc{S}(\C^2) \;,\qquad
t\mapsto \tau^1_t:=\tfrac{1+\lambda^{-1}t}{2}\,\delta_\uparrow+\tfrac{1-\lambda^{-1}t}{2}\,\delta_\downarrow \;,
$$
be the map in \eqref{eq:stateC2}.
For $x=(x_1,x_2)\in [-\lambda,\lambda]\times I$ we define a product state
$$
\tau_x:=\tau_{x_1}^1\otimes\tau_{x_2}^2 \;.
$$

\begin{prop}\label{prop:prodC2}
For all $x,y$, the states $\tau_x,\tau_y$ satisfy Pythagoras equality.
\end{prop}

\begin{proof}
Due to \eqref{eq:ineqDtimes}, it is enough to prove that
\begin{equation}\label{eq:enoughtoprove}
d_D(\tau_x,\tau_y)\leq
\sqrt{d_{D_1}(\tau^1_{x_1},\tau^1_{y_1})^2+d_{D_2}(\tau^2_{x_2},\tau^2_{y_2})^2}=\|x-y\|_2 \;,
\end{equation}
where the latter equality follows from the definition of geodesic pair.

Let $u\in\R^2$ be the unit vector 
$$
u=\frac{x-y}{\|x-y\|_2}
$$
and $f_a(t):=\tau_{ut+y}(a)$ (for all $t$ s.t.~$u_1t+y_1\in[-\lambda,\lambda]$). For all $a=(a_\uparrow,a_\downarrow)\in\A^{\mathrm{sa}}$:
\begin{equation}\label{eq:txty}
\tau_x(a)-\tau_y(a)=f_a(\|x-y\|_2)-f_a(0)=\int_0^{\|x-y\|_2}f_a'(t)\de t \leq \|f_a'\|_\infty \|x-y\|_2 \;,
\end{equation}
where $f_a'$ is the total derivative of $f_a$ with respect to $t$. Using Cauchy-Schwarz inequality in $\R^2$ we derive:
$$
f_a'(t)=\left.\left(u_1\frac{\partial}{\partial\xi_1}\tau_\xi(a)+u_2\frac{\partial}{\partial\xi_2}\tau_\xi(a)\right)\right|_{\xi=ut+y} \hspace{-3pt} \leq\left.\sqrt{
\left|\frac{\partial}{\partial\xi_1}\tau_\xi(a)\right|^2+\left|\frac{\partial}{\partial\xi_2}\tau_\xi(a)\right|^2
}\right|_{\xi=ut+y}  \hspace{-3pt} .
$$
An explicit computation gives
$$
\frac{\partial}{\partial\xi_1}\tau_\xi(a)=\frac{1}{2\lambda}\tau_{\xi_2}^2(a_\uparrow-a_\downarrow) \;,\qquad
\frac{\partial}{\partial\xi_2}\tau_\xi(a)=\tau_\xi\big((\id\otimes\delta_2)(a)\big) \;.
$$
Since $\tau_\xi(a)$ is linear in $\xi_1$, the supremum
of $\left|\tau_\xi\big((\id\otimes\delta_2)(a)\big)\right|^2$ over
all $\xi_1\in[-\lambda,\lambda]$ is attained on one of the extremal points (is a degree $2$ polynomial in $\xi_1$ with non-negative leading coefficient). Hence:
$$
\left|\tau_\xi\big((\id\otimes\delta_2)(a)\big)\right|^2
\leq\max\big\{
|\tau_{\xi_2}^2(\delta_2a_\uparrow)|^2,
|\tau_{\xi_2}^2(\delta_2a_\downarrow)|^2
\big\} \;.
$$
For any state $\varphi$ and any normal operator $b$, $|\varphi(b)|^2\leq\varphi(b^*b)$ (Kadison-Schwarz inequality).
Thus
$$
|f'_a(t)|^2\leq\max_{i\in\{\uparrow,\downarrow\}}
\tau_{\xi_2}^2\big(c_i
\big)
$$
where, for $i\in\{\uparrow,\downarrow\}$:
$$
c_i:=(\delta_2a_i)^2+
\tfrac{1}{(2\lambda)^2}(a_\uparrow-a_\downarrow)^2 \;.
$$
From \eqref{eq:sg} we get
$$
|f'_a(t)|^2\leq\max_{i\in\{\uparrow,\downarrow\}}
\big\|[D_2,a_i]^*[D_2,a_i]+\tfrac{1}{(2\lambda)^2}(a_\uparrow-a_\downarrow)^2 \big\| \;.
$$
On the other hand:
\begingroup
\addtolength{\arraycolsep}{-30pt}
$$
[D,a]^*[D,a]=\left[\;\begin{matrix}
[D_2,a_\uparrow]^*[D_2,a_\uparrow]+(2\lambda)^{-1}(a_\uparrow-a_\downarrow)^2 &
\cdot \\ \cdot &
[D_2,a_\downarrow]^*[D_2,a_\downarrow]+(2\lambda)^{-1}(a_\uparrow-a_\downarrow)^2 
\end{matrix}\;\right] ,
$$
\endgroup
where the off-diagonal elements are omitted.
By considering vectors in $\HH=\C^2\otimes\HH_2$ with one component equal to zero, we get a lower bound on the norm:
$$
\|[D,a]\|^2\geq
\max_{i\in\{\uparrow,\downarrow\}}
\big\|[D_2,a_i]^*[D_2,a_i]+\tfrac{1}{(2\lambda)^2}(a_\uparrow-a_\downarrow)^2 \big\|
\geq |f'_a(t)|^2 \;.
$$
Now \eqref{eq:txty} becomes
$\tau_x(a)-\tau_y(a)\leq\|[D,a]\|\cdot\|x-y\|_2$, that is \eqref{eq:enoughtoprove}.
\end{proof}

\subsection{From two-sheeted to \texorpdfstring{$N$}{N}-sheeted noncommutative spaces}
Inspired by \S\ref{sec:CN}, we shall now see how one can transfer results about Pythagoras from a product with $\C^2$ to a product with an arbitrary finite metric space.

Let $(\A_1,\HH_1,D_1,\gamma_1)$ be the canonical spectral triple on a metric space with $N$ points, introduced in \S\ref{sec:finite},
$(\A_2,\HH_2,D_2)$ an arbitrary unital spectral triple, and
$(\A,\HH,D)$ their product \eqref{eq:prod}.
Let $(\A',\HH',D')$ be the product of the spectral triple on a two-point space of \S\ref{sec:twopoints} with the same spectral triple $(\A_2,\HH_2,D_2)$ above. Let $\delta_\uparrow(a_\uparrow,a_\downarrow)=a_\uparrow$
and
$\delta_\downarrow(a_\uparrow,a_\downarrow)=a_\downarrow$
be the two pure states of $\C^2$.

\begin{prop}\label{prop:C2toCN}
Let $\varphi_2,\psi_2\in\mc{S}(\A_2)$. If Pythagoras equality is satisfied by $\delta_\uparrow\otimes\varphi_2$ and $\delta_\downarrow\otimes\psi_2$
in the product $(\A',\HH',D')$,
then it is satisfied by
$\varphi_1\otimes\varphi_2$ and $\psi_1\otimes\psi_2$
in the product $(\A,\HH,D)$ for any two pure states
$\varphi_1,\psi_1$ of $\A_1=\C^N$.
\end{prop}

\begin{proof}
We have a decomposition 
$\HH=\bigoplus_{i\neq j}\C^2\otimes\HH_2$, where on the summand $(i,j)$ the representation is given by
$$
\pi_{ij}(a)=\left[\!\begin{array}{cc} a_i & 0 \\ 0 & a_j \end{array}\!\right] \;,
$$
for all $a=(a_1,\ldots,a_N)\in \A=(\A_2)^N\simeq\C^N\otimes\A_2$ (each $a_i$ here is a bounded operator on $\HH_2$). The Dirac operator is
$D=\bigoplus_{i\neq j}D_{ij}$, with
$$
D_{ij}=\frac{1}{g_{ij}}\left[\begin{array}{rr} 0 & 1 \\ 1 & 0 \end{array}\right]+
\left[\begin{array}{cr} D_2 & 0 \\ 0 & \!\!-D_2 \end{array}\right]
\;,
$$
where each matrix entry is an operator on $\HH_2$. Let $\varphi_1=x_r$ and $\psi_1=x_s$ be two pure states as in \eqref{eq:xxstar}.
If $r=s$, Pythagoras equality is trivial (it follows from Theorem \ref{thm:2.3}(i) and Lemma \ref{lemma:2.4}, i.e.~the fact that taking a product doesn't change the horizontal resp.~vertical distance). We can then assume that $r\neq s$.

Since $\|[D,a]\|=\max_{i\neq j}\|[D_{ij},\pi_{ij}(a)]\|\geq\|[D_{rs},\pi_{rs}(a)]\|$,
then
$$
d_D(\varphi,\psi)\leq d_{D_{rs}}(\varphi,\psi)
=\sup\big\{\varphi_2(a_r)-\psi_2(a_s):
\|[D_{rs},\pi_{rs}(a)]\|\leq 1
\big\} \;.
$$
On the other hand, $d_{D_{rs}}(\varphi,\psi)$ is the distance on a product of $(\A_2,\HH_2,D_2)$ with a two-point space
(the two states $x_r$ and $x_s$ will give one $\delta_\uparrow$ and one $\delta_\downarrow$), which by hypothesis is equal to
$$
\sqrt{g_{rs}^2+d_{D_2}(\varphi_2,\psi_2)^2}=
\sqrt{d_{D_1}(x_r,x_s)^2+d_{D_2}(\varphi_2,\psi_2)^2} \;.
$$
This proves $d_D(\varphi,\psi)\leq
d_{D_1}\!\boxtimes d_{D_2}(\varphi,\psi)$, the opposite inequality being
always true, cf.~\eqref{eq:ineqDtimes}.
\end{proof}

\noindent
Note that the states $\varphi_2,\psi_2$ in previous proposition are not necessarily pure.

\begin{cor}\label{cor:7.3}
Pythagoras equality is satisfied by pure states in the product of any Riemannian manifold with
any finite metric space \cite{MW02}, and by translated states in the product of Moyal spectral triple with
any finite metric space \cite{MT13}.
\end{cor}


\section{A digression on purification}\label{sec:pur}

Consider a product of unital spectral triples $(\A_1,\HH_1,D_1,\gamma_1)$ and $(\A_2,\HH_2,D_2)$. From \eqref{eq:ineqDtimes}:
$$
d_{D_1}(\varphi_1,\psi_1)\leq d_D(\varphi,\psi)
$$
for any states $\varphi,\psi$ with marginals $\varphi_1^\flat=\varphi_1$ and $\psi_1^\flat=\psi_1$.
When the second spectral triple is given by $\A_2=\HH_2=\C$ and $D_2=0$, for $\varphi=\varphi_1\otimes\id$ and $\psi=\psi_1\otimes\id$ (the identity map is a state on $\C$), the above inequality is clearly an equality. We can then say that:
$$
d_{D_1}(\varphi_1,\psi_1)=\min\,d_D(\varphi,\psi)
$$
is a minimum over all the \emph{extensions} $\varphi,\psi$ of 
$\varphi_1,\psi_1$ and over all unital spectral triples $(\A_2,\HH_2,D_2)$. This may seem an overly convoluted way to look at the spectral distance, but establishes an analogy with the purified distance of quantum information, which is briefly discussed below.

\smallskip

Given two states $\varphi_1\in\mc{S}(A_1)$ and $\varphi_2\in\mc{S}(A_2)$, a \emph{purification} is any pure state $\varphi\in\mc{S}(A_1\otimes A_2)$ with marginals $\varphi_1$ and $\varphi_2$.
In Example \ref{ex:2.1}, the Bell state is a purification of the tracial states $\varphi_1=\varphi_2=\tr_{\C^2}(\,.\,)$.

\smallskip

Let us focus on the example of matrix algebras, which is the one studied in quantum information. States are in bijection with density matrices, and in this section they will be identified. A state is pure if the density matrix $\rho\in M_{N}(\C)$ is a projection in the direction of a unit vector $v$ of $\C^N$, i.e.~$\rho=vv^*$, where we think of $v$ resp.~$v^*$ as a column resp.~row vector.

In the matrix case, every state $\rho$ has an almost canonical purification. Since $\rho$ is positive, we can find an orthonormal eigenbasis $\{v_i\}_{i=1}^N$ of $\C^N$ and write
$$
\rho=\sum\nolimits_{i=1}^Np_iv_iv_i^* \;,
$$
where $\{p_i\}_{i=1}^N$ is the probability distribution given by the eigenvalues of $\rho$. A purification is then the state $\rho^\circ\in M_{N}(\C)\times M_{N}(\C)$ given by
$$
\rho^\circ=\sum\nolimits_{i,j=1}^N\sqrt{\rule{0pt}{7pt}p_ip_j}\,v_iv_j^*\otimes v_iv_j^* \;,
$$
which is a projection in the direction of $v^\circ=\sum\nolimits_{i=1}^N\sqrt{\rule{0pt}{7pt}p_i}\,v_i\otimes v_i$. It is almost canonical in the sense that it depends on the choice of eigenbasis.

The \emph{purified distance} $d_{\mc{P}}$ is usually defined in terms of Uhlmann's fidelity, but the most interesting characterization is as the minimal trace distance between all possible purifications of the states \cite[\S3.4]{Tom15}. For $\rho,\tau\in M_N(\C)$:
\begin{equation}\label{eq:purified}
d_{\mc{P}}(\rho,\tau)=\inf\, \sqrt{1-\left|\inner{v,w}\right|^2}
\end{equation}
where the inf is over all rank $1$ projections $\rho'=vv^*,\tau'=ww^*\in M_N(\C)\otimes M_{N'}(\C)$, for all $N'\geq 1$, that have $\rho$ and $\tau$ as marginals:
$$
(\id_{\C^N}\otimes\tr_{\C^{N'}})(\rho')=\rho \;,\qquad
(\id_{\C^N}\otimes\tr_{\C^{N'}})(\tau')=\tau \;.
$$
One can easily verify that the trace distance and the purified distance coincide on pure states.

\begin{ex}[Bloch's sphere]
Consider the map $B_3\to\mc{S}(M_2(\C))$, $x\mapsto\rho_x$, in \S\ref{sec:Bloch}. To compute the purified distance we use the formula $d_{\mc{P}}(\rho,\tau)=\sqrt{1-F(\rho,\tau)}$
\cite[Def.~3.8]{Tom15}, together with the formula
$F(\rho,\tau)=\tr(\rho\tau)+2\sqrt{\det\rho\hspace{1pt}\det\tau}$
for the fidelity \cite[Eq.~(9.47)]{ZS01}.
We get
$$
d_{\mc{P}}(\rho_y,\rho_y)=2^{-\frac{1}{2}}
\sqrt{1-\inner{x,y}-(1-\|x\|_2^2)^{\frac{1}{2}}(1-\|y\|_2^2)^{\frac{1}{2}}} \;.
$$
If $x,y$ are unit vectors (pure states),
with some algebraic manipulation one verifies that $d_{\mc{P}}$ is proportional to the Euclidean distance: $d_{\mc{P}}(\rho_y,\rho_y)=\frac{1}{2}\|x-y\|_2$.
\end{ex}

\appendix
\section{\texorpdfstring{The surface $\Delta^1\times\Delta^1\subset\Delta^3$}{}}\label{app:A}
Since $\Delta^1$ is a segment, from a metric point of view $\Delta^1\times\Delta^1$ with product metric is a square.
We can plot its embedding in $\Delta^3$ as follows. Vertices of $\Delta^3$ are vectors $\{e_i\}_{i=1}^4$ of the canonical basis of $\R^4$; let $f:\R^4\to\R^3$ be the linear map 
sending the vertices of $\Delta^3$ to alternating vertices of a cube, that is: $f(e_1)=(1,1,1)$, $f(e_2)=(1,-1,-1)$,
$f(e_3)=(-1,1,-1)$, $f(e_4)=(-1,-1,1)$.

With the obvious identification $\Delta^1\simeq [-1,1]$,
the map from $\Delta^1\times\Delta^1$ to product states is
$$
[-1,1]^2\ni (t,s)\mapsto
p(t,s):=\left(
\tfrac{(1+t)(1+s)}{4},
\tfrac{(1+t)(1-s)}{4},
\tfrac{(1-t)(1+s)}{4},
\tfrac{(1-t)(1-s)}{4}
\right)
\in\Delta^3 .
$$
Performing a parametric plot of
\begin{equation}\label{eq:parplot}
(t,s)\mapsto f\circ p(t,s)=(t,s,t\cdot s)
\end{equation}
one gets the surface in Figure \ref{fig:parplot}. One can imagine that the square is first bent along a diagonal, so that its two halves cover two faces of the tetrahedron\footnote{Of course, one has to shrink the diagonal until its length matches the edges.}. Then it is bent by pushing its center, until it coincides with the barycenter of the tetrahedron\footnote{The center $(0,0)$ of the square is sent to $(0,0,0)\in\R^3$, which is the barycenter of $f(\Delta^3)$.}.

\begin{figure}[t]
  \centering
	\includegraphics[height=5cm]{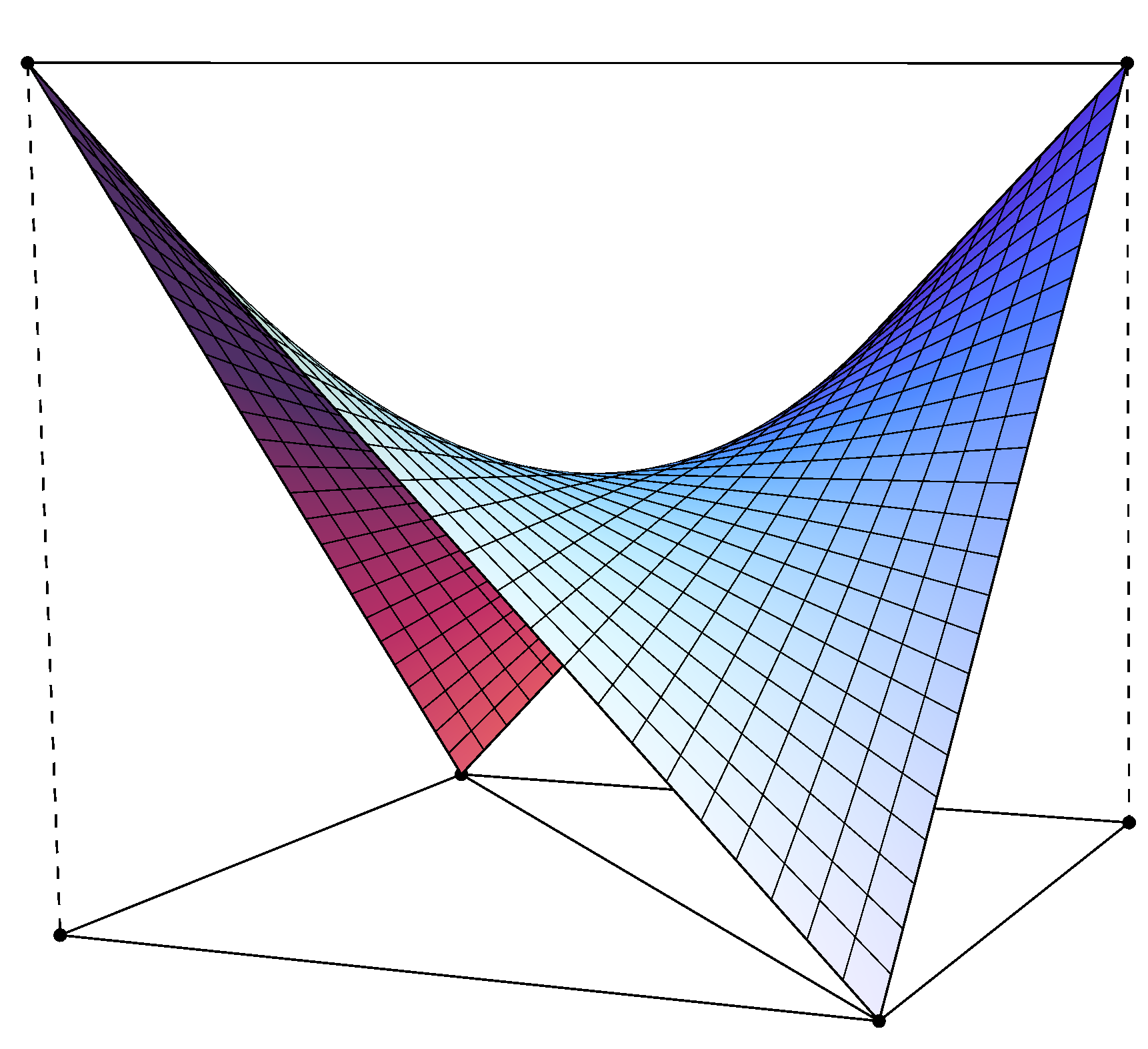}
	
	\vspace{-10pt}
	
	\caption{Plot of \eqref{eq:parplot}.}
	\label{fig:parplot}
\end{figure}

Note that this surface is just the intersection of a hyperbolic paraboloid (with equation \mbox{$z=xy$}) with the cube of vertices $(\pm 1,\pm 1,\pm 1)$, so in particular it is not an isometric deformation of the original square (by Gauss's Theorema Egregium, since it is negatively curved). Thus, the geodesic distance induced by this embedding of $\Delta^1\times\Delta^1\to\R^3$ is not the product distance. On the other hand, the embeddings of $\Delta^1\times\{s\}$ and $\{t\}\times\Delta^1$ are isometric (these are the straight lines with constant $s$ resp.~$t$ which are shown in the figure). The product metric is the pullback of the Euclidean distance on the base square in Figure \ref{fig:parplot} via the vertical projection.

One can also check that the map $\varphi\to\varphi^\flat$ from a state to the product of its marginals is just the vertical projection to the red-blue surface of the corresponding point in the tetrahedron in Figure \ref{fig:parplot}.

\end{document}